\documentclass[12pt,usenames,dvipsnames]{article}

\usepackage{latexsym}
\usepackage{amssymb,amsfonts,amsmath}
\usepackage{graphicx} 
\usepackage{indentfirst}
\usepackage{bbm}
\usepackage{amssymb}
\usepackage{verbatim}
\usepackage{amsmath, amsthm,amssymb}
\usepackage{mathrsfs}
\usepackage{amsfonts}
\usepackage{dsfont}
\usepackage{cite}
\usepackage{xcolor}
\usepackage{enumerate}
\usepackage{enumitem}

\parskip=\medskipamount

\usepackage{tikz,tikz-cd}
\usepackage{tensor}
\usepackage{soul}
\usepackage{amsmath,amssymb,calc, amsthm,bbm, epsfig,psfrag, mathtools}
\newtheorem{theorem}{Theorem}
\newtheorem{corollary}{Corollary}
\usepackage{latexsym,bm,amsfonts}
\usepackage{graphicx, enumerate}
\usepackage[numbers,sort&compress]{natbib}
\usepackage{mathrsfs}
\allowdisplaybreaks

\newcommand{\Comment}[1]{{}}
\definecolor{darkblue}{rgb}{0.15,0.35,0.55}
\definecolor{reddish}{rgb}{0.65, 0.2, 0.2}
\usepackage[linktocpage=true]{hyperref}
\hypersetup{
colorlinks=true,
citecolor=darkblue,
linkcolor=reddish,
urlcolor=darkblue,
pdfauthor={},
pdftitle={},
pdfsubject={}
}
\usepackage{cleveref}

\makeatletter
\renewcommand\section{\@startsection {section}{1}{\z@}%
                                   {-3.5ex \@plus -1ex \@minus -.2ex}
                                   {2.3ex \@plus.2ex}%
                                   {\normalfont\large\bfseries}}
\renewcommand\subsection{\@startsection{subsection}{2}{\z@}%
                                     {-3.25ex\@plus -1ex \@minus -.2ex}%
                                     {1.5ex \@plus .2ex}%
                                     {\normalfont\bfseries}}
\makeatother

\let\non\nonumber

\newcommand{\bea}{\begin{eqnarray}}
\newcommand{\eea}{\end{eqnarray}}

\newcommand{\be}{\begin{equation}}
\newcommand{\ee}{\end{equation}}

\newcommand{\bma}{\begin{pmatrix}}
\newcommand{\ema}{\end{pmatrix}}

\newcommand{\bsubeq}{\begin{subequations}}
\newcommand{\esubeq}{\end{subequations}}
\newcommand{\bsubea}{\begin{subequations}\bea}
\newcommand{\esubea}{\eea\end{subequations}}

\newfont{\goth}{ygoth.tfm scaled 1200}                   

 \numberwithin{equation}{section}

\def\1{{(1)}}
\def\2{{(2)}}
\def\3{{(3)}}

\newcommand{\overbar}[1]{\mkern 1.5mu\overline{\mkern-1.5mu#1\mkern-1.5mu}\mkern 1.5mu}

\def\TT{{T\overbar{T}}}

\def\Gt{{\widetilde{G}}}
\def\Ft{{\widetilde{F}}}

\def\s{\sigma}

\def\a{\alpha}
\def\b{\beta}

\def\d{\delta}

\def\g{\gamma}

\def\l{\lambda}
\def\m{\mu}
\def\n{\nu}

\def\s{\sigma}

\def\Q{\Theta}

\def\U{\Upsilon}

\newcommand{\ad}{{\dot\alpha}}

\newcommand{\pa}{\partial}

\newcommand{\hf}{\frac12}

\newcommand {\cN}{{\cal N}}

\newcommand{\cO}{{\mathcal O}}

\newcommand{\cL}{{\mathcal L}}
\newcommand{\cE}{{\mathcal E}}
\newcommand{\cH}{{\mathcal H}}
\newcommand{\sE}{{\mathscr{E}}}
\newcommand{\sH}{{\mathscr{H}}}

\DeclareMathOperator{\sech}{sech}

\renewcommand{\be}{\begin{equation}}
\renewcommand{\ee}{\end{equation}}
\newcommand{\bpm}{\begin{pmatrix}}
\newcommand{\epm}{\end{pmatrix}}

\newcommand{\beqn}{\begin{eqnarray}}
\newcommand{\eeqn}{\end{eqnarray}}

\usepackage{slashed}

\DeclareMathOperator{\tr}{tr}

\begin{document}

\begin{titlepage}
\begin{flushright}
\today \\
\end{flushright}
\vspace{5mm}

\begin{center}
{\Large \bf 
Duality-Invariant Non-linear Electrodynamics \\
and Stress Tensor Flows 
}
\end{center}

\begin{center}

{\bf
Christian Ferko,${}^{a}$
Sergei M. Kuzenko,${}^{b}$
Liam Smith,${}^{c}$ and\\
Gabriele Tartaglino-Mazzucchelli${}^{c}$
} \\
\vspace{5mm}
\footnotesize{
${}^{a}$
{\it 
Center for Quantum Mathematics and Physics (QMAP), 
\\ Department of Physics \& Astronomy,  University of California, Davis, CA 95616, USA
}
 \\~\\
${}^{b}$
{\it 
Department of Physics M013, The University of Western Australia
\\
35 Stirling Highway, Perth W.A. 6009, Australia
}
 \\~\\
${}^{c}$
{\it 
School of Mathematics and Physics, University of Queensland,
\\
 St Lucia, Brisbane, Queensland 4072, Australia}
 }
\vspace{2mm}
~\\
\texttt{caferko@ucdavis.edu, 
sergei.kuzenko@uwa.edu.au, 
liam.smith1@uq.net.au, 
g.tartaglino-mazzucchelli@uq.edu.au}\\
\vspace{2mm}

\end{center}

\begin{abstract}
\baselineskip=14pt

\noindent  
Given a model for self-dual non-linear electrodynamics in four spacetime dimensions, any deformation of this theory which is constructed from the duality-invariant energy-momentum tensor preserves duality invariance. In this work we present new proofs of this known result, and also establish a previously unknown converse: any parameterized family of duality-invariant Lagrangians, all constructed from an Abelian field strength $F_{\mu \nu}$ but not its derivatives, is related by a generalized stress tensor flow, in a sense which we make precise. We establish this and other properties of stress tensor deformations of theories of non-linear electrodynamics using both a conventional Lagrangian representation and using two auxiliary field formulations. We analyze these flows in several examples of duality-invariant models including the Born-Infeld and ModMax theories, and we derive a new auxiliary field representation for the two-parameter family of ModMax-Born-Infeld theories. These results suggest that the space of duality-invariant theories may be characterized as a subspace of theories of electrodynamics with the property that all tangent vectors to this subspace are operators constructed from the stress tensor.

\end{abstract}
\vspace{5mm}

\vfill
\end{titlepage}

\newpage
\renewcommand{\thefootnote}{\arabic{footnote}}
\setcounter{footnote}{0}

\tableofcontents{}
\vspace{1cm}
\bigskip\hrule


\allowdisplaybreaks

\section{Introduction} 
\label{sec:intro}

A deeper understanding of the phenomenon of duality has been a remarkable source of progress in theoretical physics. Broadly speaking, a duality is any correspondence in which there exist two -- seemingly different -- descriptions of the same physical system. 

One general mechanism by which such correspondences emerge is strong-weak duality. This term often refers to the S-duality of type IIB string theories \cite{Schwarz:1993vs,Sen:1994yi,Schwarz:1995dk} in which the axio-dilaton $\tau = C_0 + \frac{i}{g_s}$ transforms via an $SL( 2 , \mathbb{Z} )$ transformation; the closely related Montonen-Olive duality \cite{Montonen:1977sn} involves a similar transformation on the complex coupling $\tau = \frac{\theta}{2 \pi} + \frac{4 \pi i}{g^2}$ in $4d$ supersymmetric gauge theories. This class of strong-weak or electric-magnetic dualities generalize the electromagnetic duality of Maxwell's equations, which form the simplest and earliest example within this class, and which will be the focus of the present work.

The basic observation of electromagnetic duality is that, in the presence of both electric sources $j_{\rm e}^\mu$ and magnetic sources $j_{\rm m}^\mu$, the equations of motion for the Maxwell theory are
\begin{align}\label{maxwell_with_sources}
    \partial_\nu F^{\mu \nu} = j_{\rm e}^\mu \, , \qquad \partial_\nu \widetilde{F}^{\mu \nu} = j_{\rm m}^\mu \, , 
\end{align}
where $\widetilde{F}^{\mu \nu} = \frac{1}{2} \epsilon^{\mu \nu \rho \sigma} F_{\rho \sigma}$ is the Hodge dual of $F_{\mu \nu}$. The equations (\ref{maxwell_with_sources}) are invariant under the simultaneous replacements
\begin{align}\label{duality_sources}
    F^{\mu \nu} \to \widetilde{F}^{\mu \nu} \, , \quad \widetilde{F}^{\mu \nu} \to - F^{\mu \nu} \, , \quad j_{\rm e}^\mu \to j_{\rm m}^\mu \, , \quad j_{\rm m}^\mu \to - j_{\rm e}^\mu \, .
\end{align}
This duality transformation (\ref{duality_sources}) exchanges both electric and magnetic fields, along with electric and magnetic sources. For instance, point electric charges are traded for magnetic monopoles, and vice-versa, under this map. This makes it straightforward to see why such a transformation is also referred to as a strong-weak duality. By the Dirac quantization condition, the magnetic coupling constant is the inverse of the electric coupling; the latter is the usual fine structure constant. Thus we conventionally think of an electric charge as a weakly coupled particle and a magnetic monopole as a strongly coupled soliton. The duality (\ref{duality_sources}) therefore interchanges a weak-coupling object with a strong-coupling object.

In general, a duality relates a pair of descriptions in two \emph{different} theories. Because the couplings are part of the data that defines a physical theory, the strong-weak duality exchanging electrically charged particles and magnetic monopoles can be viewed as a correspondence beween a theory with coupling $g$ and a theory with coupling $\frac{1}{g}$.\footnote{Likewise, the Montonen-Olive duality of super-Yang-Mills relates a theory with one choice of the coupling $g$ and theta angle $\theta$ to a theory with different values of these two parameters \cite{Witten:1978mh, Osborn:1979tq}.}

However, in special cases a duality transformation relates two instances of the \emph{same} physical theory. Such a theory is said to be self-dual. One example is the vacuum Maxwell theory, which corresponds to the equations of motion (\ref{maxwell_with_sources}) with $j_{\rm e}^\mu = j_{\rm m}^\mu = 0$. In this case, there are no coupling constants for either electrically charged particles or magnetically charged monopoles, and thus the duality transformation (\ref{duality_sources}) simply exchanges the electric and magnetic fields with no further modifications.

Self-duality is a form of enhanced symmetry that a particular theory might enjoy which imposes additional constraints. For instance, the electric-magnetic duality of the Maxwell theory implies a certain statement of helicity conservation \cite{Calkin}. A great deal of previous work has been devoted to studying the self-duality of theories of non-linear electrodynamics; see for instance \cite{Gaillard:1981rj,Gibbons:1995cv,Gaillard:1997rt,Gaillard:1997zr,Hatsuda:1999ys,Kuzenko:2000uh,Ivanov:2002ab,Ivanov:2003uj} and references therein. It is therefore of great interest to characterize which other theories exhibit self-duality, and to better understand the interplay between self-duality and other properties.

More precisely, by ``self-dual non-linear electrodynamics'' we understand $U(1)$ duality-invariant non-linear extensions of Maxwell's theory. Self-duality under $U(1)$ duality rotations implies self-duality under a Legendre transformation \cite{Gaillard:1997rt}.
In order for a theory with Lagrangian $\cL (F)$ to possess $U(1)$ duality invariance, 
the Lagrangian must satisfy the so-called self-duality equation\footnote{The terminology ``self-duality equation'' was introduced by Gaillard and Zumino \cite{Gaillard:1997rt}.}
\cite{B-B, Gibbons:1995cv,Gaillard:1997rt,Gaillard:1997zr}
\bea
F^{\m\n}  \widetilde{F}_{\m\n} + G^{\m\n}  \widetilde{G}_{\m\n} =0~, \qquad
\Gt_{\mu \nu} = 2 \frac{\partial \mathcal{L}}{\partial F^{\mu \nu}} ~ .
\label{SDequation}
\eea
The formalism of \cite{Gaillard:1981rj,Gibbons:1995cv,Gaillard:1997rt,Gaillard:1997zr}
was extended to duality-invariant theories 
with higher derivatives\footnote{Further aspects of duality-invariant theories 
with higher derivatives were studied, e.g., in \cite{AFZ,Chemissany:2011yv,AF,AFT}.} 
 \cite{Kuzenko:2000uh}, as well as to the case of general $U(1)$ duality-invariant $\cN=1$ and $\cN=2$  supersymmetric theories  \cite{Kuzenko:2000tg, Kuzenko:2000uh}. For a comprehensive review of these and related developments, see \cite{Kuzenko:2000uh,AFZ}.
In this paper our analysis is restricted to self-dual models for non-linear electrodynamics without higher derivatives.

Quite generally, a useful way to understand any desirable feature of a physical system is to study its behavior under deformations. For instance, one might begin with a self-dual theory of electrodynamics such as Maxwell -- we will also refer to such theories as duality-invariant -- and ask whether the property of duality-invariance is preserved under some class of deformations.

This brings us to the second broad topic of this work, which is deformations of field theories that are constructed from the energy-momentum tensor. At the classical level, we define such a deformation via a differential equation of the form
\begin{align}\label{flow_general}
    \frac{\partial \mathcal{L}^{(\lambda)}}{\partial \lambda} = \mathcal{O} \left( T_{\mu \nu}^{(\lambda)};\l \right) \, , 
\end{align}
where the object $\mathcal{O} \big( T_{\mu \nu}^{(\lambda)};\l \big)$ is any Lorentz scalar constructed from the Hilbert stress tensor\footnote{We will use the terms ``energy-momentum tensor,'' ``stress-energy tensor,'' and ``stress tensor'' interchangeably to refer to this object.} associated with the theory $\mathcal{L}^{(\lambda)}$. The latter is defined by
\begin{align}
    T_{\mu \nu}^{(\lambda)} = - 2 \frac{\partial \mathcal{L}^{(\lambda)}}{\partial g^{\mu \nu}} + g_{\mu \nu} \mathcal{L}^{(\lambda)} \, .
\end{align}

Beginning from an initial condition $\mathcal{L}^{(\lambda = 0)} = \mathcal{L}_0$, which we refer to as the seed theory, the solution to the differential equation (\ref{flow_general}) produces a one-parameter family of Lagrangians labeled by a flow parameter $\lambda$.

The most famous flow equation of this form is the $\TT$ deformation of two dimensional quantum field theories, which was introduced in \cite{Zamolodchikov:2004ce} and further explored in \cite{Smirnov:2016lqw,Cavaglia:2016oda}. This $\TT$ operator, which in two dimensions is proportional to the determinant of the energy-momentum tensor, has the remarkable property that it can be used to define not only a classical flow equation for the Lagrangian, but even a fully quantum mechanical deformation of a $2d$ QFT. The definition of the quantum $\TT$ deformation relies on the fact that the coincident point limit
\begin{align}\label{TT_defn}
    \mathcal{O}_{\TT} ( x ) = \lim_{y \to x} \left( T^{\mu \nu} ( x ) T_{\mu \nu} ( y ) - \tensor{T}{^\mu_\mu} ( x ) \tensor{T}{^\nu_\nu} ( y ) \right) \, , 
\end{align}
defines a local operator in the spectrum of any translation-invariant two-dimensional quantum field theory, up to total derivative ambiguities, as shown in \cite{Zamolodchikov:2004ce}. 

Although the combination of stress tensors appearing in (\ref{TT_defn}) has dimension $4$, and is thus irrelevant in the Wilsonian sense, surprisingly this deformation is still solvable in that one can often compute quantities in the deformed theory at finite $\lambda$. Examples include the deformed finite-volume spectrum \cite{Smirnov:2016lqw,Cavaglia:2016oda}, $S$-matrix \cite{Dubovsky:2017cnj}, and torus partition function \cite{Cardy:2018sdv,Datta:2018thy,Aharony:2018bad}; each of these observables admits some relation between the quantity in the deformed theory at finite $\lambda$ and the seed theory at $\lambda = 0$. Another property of the $\TT$ flow is that it often preserves symmetries and other desirable features of the seed theory, such as integrability \cite{Smirnov:2016lqw,Chen:2021aid} and supersymmetry \cite{Baggio:2018rpv,Chang:2018dge,Jiang:2019hux,Chang:2019kiu,Coleman:2019dvf,Ferko:2019oyv,Ebert:2022xfh,Lee:2021iut,Lee:2023uxj}. See \cite{Jiang:2019epa} for a review of other results concerning $\TT$ deformations.

In spacetime dimensions $d > 2$, it is not known how to define an analogue of the local $\TT$ operator at the quantum level; discussions of possible generalizations can be found in \cite{Taylor:2018xcy,Bonelli:2018kik}. However, one might hope to find clues about potentially interesting operators by investigating purely classical flows for the Lagrangian which take the form (\ref{flow_general}). One reason to expect that this might be useful is that the analogous classical flows in $d = 2$ also exhibit interesting structures. For instance, the classical $\TT$ flow equation deforms the seed theory of a single free scalar field in $d = 2$ into the theory of a gauge-fixed Nambu-Goto string in a three-dimensional target space \cite{Cavaglia:2016oda}. Likewise, in four spacetime dimensions, the classical flow equation
\begin{align}\label{TTbar_flow}
    \frac{\partial \mathcal{L}^{(\lambda)}}{\partial \lambda} = \frac{1}{8} \left( T^{\mu \nu} T_{\mu \nu} - \frac{1}{2} \left( \tensor{T}{^\mu_\mu} \right)^2 \right) \, , 
\end{align}
with a seed theory corresponding to the Maxwell Lagrangian, $\mathcal{L}_0 = - \frac{1}{4} F_{\mu \nu} F^{\mu \nu}$, has a solution which is the Born-Infeld theory describing the effective gauge dynamics on a brane \cite{Conti:2018jho}. This is a hint that stress tensor deformations appear to be related to theories of strings and branes.\footnote{There is another connection between little string theory and the single trace $\TT$ operator of \cite{Giveon:2017nie,Asrat:2017tzd,Giveon:2017myj}, whose properties such as the deformed spectrum can be understood holographically via a gravity analysis \cite{Apolo:2019zai,Chang:2023kkq}.
}

A similar classical flow equation can be defined which deforms the Maxwell theory into the Born-Infeld theory in $d = 3$, or which deforms a free scalar into the Nambu-Goto action in any spacetime dimension \cite{Ferko:2023sps}. However, these more general flow equations require a new ingredient: one must also introduce an object of the form
\begin{align}\label{general_rTT}
    \mathcal{R} = \sqrt{ \frac{1}{d} T^{\mu \nu} T_{\mu \nu} - \frac{1}{d^2} \left( \tensor{T}{^\mu_\mu} \right)^2 } \, .
\end{align}
When $d = 2$, this combination (\ref{general_rTT}) reduces to the root-$\TT$ operator introduced in \cite{Ferko:2022cix}; related work can be found in \cite{Rodriguez:2021tcz,Bagchi:2022nvj,Tempo:2022ndz}. Unlike the irrelevant $\TT$ operator, the root-$T^2$ operator $\mathcal{R}$ is classically marginal. It appears to enjoy some of the desirable features of the $\TT$ deformation, such as preserving classical integrability for certain $2d$ models \cite{Borsato:2022tmu}, although it is not known whether the $2d$ root-$\TT$ operator can be defined at the quantum level.\footnote{A proposed flow equation for the finite-volume spectrum of a $2d$ CFT deformed by root-$\TT$, which would represent a quantum result, was presented in \cite{Ebert:2023tih} based on a holographic analysis similar to that of \cite{Guica:2019nzm,Ebert:2022ehb}.} However, our primary motivation for studying the combination (\ref{general_rTT}) is that it can be used to build flow equations which lead to interesting classical actions. For instance, solving the flow equation
\begin{align}\label{root_TTbar_flow}
    \frac{\partial \mathcal{L}^{({\gamma)}}}{\partial \gamma} = \frac{1}{2} \sqrt{ T^{\mu \nu} T_{\mu \nu} - \frac{1}{4} \left( \tensor{T}{^\mu_\mu} \right)^2 } \, ,
\end{align}
with a Maxwell seed, which is a deformation by $\mathcal{R}$ in $d = 4$, gives a solution,
\begin{align}\label{modmax_lagrangian}
    \mathcal{L}_{\text{ModMax}} = 
    - \frac{1}{4} \cosh ( \gamma ) F^{\mu \nu} F_{\mu \nu}
    + \frac{1}{4} \sinh ( \gamma ) 
    \sqrt{ 
    \left( F^{\mu \nu} F_{\mu \nu} \right)^2 + \left( F^{\mu \nu} \widetilde{F}_{\mu \nu} \right)^2
    } \, ,
\end{align}
which is the Modified Maxwell or ModMax theory introduced in \cite{Bandos:2020jsw}. This ModMax theory is of considerable interest because it is the unique conformally invariant and electromagnetic duality-invariant extension\footnote{The program to combine $U(1)$ duality invariance with $\cN=2$ superconformal symmetry was put forward in 2000  \cite{Kuzenko:2000tg}. It was completed in \cite{Kuzenko:2021cvx}, where the $\cN=2$ superconformal $U(1)$ duality-invariant model was proposed to describe the low-energy effective action for $\cN=4$ super-Yang-Mills theory. In the $\cN=0$ and $\cN=1$ cases, non-linear $U(1)$ duality-invariant (super)conformal theories do not possess a weak field limit.}
 of the $4d$ Maxwell theory.\footnote{See \cite{Sorokin:2021tge} for an instructive set of lectures on theories of non-linear electrodynamics, including ModMax. } Several related ModMax-like theories have also been studied, including a supersymmetric extension \cite{Bandos:2021rqy,Kuzenko:2021cvx}, a two-parameter family of ModMax-Born-Infeld theories and $6d$ tensor analogues \cite{Bandos:2020hgy}, a $(0+1)$-dimensional ModMax-like harmonic oscillator \cite{Garcia:2022wad,Ferko:2023ozb,Ferko:2023iha}, and a supersymmetric non-linear sigma model  whose Lagrangian has a structure similar to that of ModMax \cite{Kuzenko:2023ysh}.\footnote{This duality-invariant supersymmetric $\s$-model is known as  the MadMax $\s$-model  \cite{Kuzenko:2023ysh}.} 

The relationship between stress tensor flows and these various theories of non-linear electrodynamics has, to some degree, already been explored in several works \cite{Babaei-Aghbolagh:2020kjg,Babaei-Aghbolagh:2022uij,Babaei-Aghbolagh:2022itg,Ferko:2022iru,Ferko:2023ruw}. However, one point merits further investigation, which brings us back to our preceding discussion on duality invariance. All of the theories of electrodynamics that we have discussed here -- Born-Infeld, ModMax, and ModMax-Born-Infeld -- are special insofar as they are invariant under electric-magnetic duality transformations. One might have expected this property because all of these theories can be realized as stress tensor deformations of the Maxwell theory. Because the Maxwell theory is electromagnetic duality invariant, and the energy-momentum tensor of a self-dual theory is also a duality-invariant quantity, it seems natural that any stress tensor flow will also preserve duality invariance. Indeed this is the case, as was pointed out in \cite{Ferko:2023ruw} and will be reviewed in the present work.

This motivates a more detailed study of the relationship between the two topics that we have discussed in this introduction, namely duality invariance and stress tensor deformations. For example, one might ask whether every duality-preserving deformation of a self-dual theory of electrodynamics is also a stress tensor deformation. We will see that this is the case, at least for theories without higher-derivative interactions. It is also natural to wonder whether the interplay between stress tensor flows and duality invariance can be made more transparent using an auxiliary field formulation which makes self-duality manifest, and we will explore this topic as well. Together these results paint a picture which suggests a deeper connection between deformations driven by conserved quantities and various notions of self-duality, and one might hope that some of these insights generalize to other instances of strong-weak duality.

The layout of this paper is as follows. In Section \ref{sec:duality_and_flows}, we review various properties of $\TT$ flows in $4d$ duality-invariant theories of electrodynamics, and prove that deformations of such theories by duality-invariant functions (such as those constructed from the stress tensor) preserve duality invariance. Section \ref{sec:auxreview} reviews the two auxiliary field formulations, referred to as the $\nu$ and $\mu$ representations, which were introduced by Ivanov and Zupnik in \cite{Ivanov:2001ec} and that we employ in this paper. In Section \ref{sec:stresstensor}, we obtain expressions for components of the stress tensor of duality-invariant theories in the $\nu$ and $\mu$ representations; these expressions can be used to define generic flow equations. Section \ref{sec:generalargument} shows that parameterized families of duality-invariant theories in the auxiliary field representations satisfy stress tensor flow equations ``almost everywhere'' (that is, away from a set of measure zero). We collect several examples of flows for duality-invariant theories in Section \ref{sec:examples}, and present a new $\mu$-frame definition of the ModMax-Born-Infeld theory. Finally, in Section \ref{sec:conclusion} we conclude and identify directions for future research. The details of various technical computations have been included in Appendix \ref{Appendix-A}.

\section{Self-dual non-linear electrodynamics and $\TT$-like flows}\label{sec:duality_and_flows}

In this section we consider a generic theory of non-linear electrodynamics described by a Lagrangian $\cL=\cL(F_{\mu\nu})$ with $F_{\mu\nu}=(\pa_{\mu}A_\nu-\pa_{\nu}A_\mu)$ being the field strength for an Abelian gauge field $A_\mu$. Note that we do not consider higher-derivative Lagrangians where $\cL$ could have functional dependence on derivatives of $F_{\mu\nu}$. One of the main aims of our paper is to understand how electric-magnetic duality invariance behaves in general under the flow equation (\ref{flow_general}).
Our analysis links this problem to $\TT$-like flows.


\subsection{Generalities}
\label{sec:generalities}

Generic models of our interest can be parametrised in terms of Lorentz invariant Lagrangians of the form $\cL=\cL(S,P)$ with\footnote{ Gaillard and Zumino \cite{Gaillard:1997rt} worked with the invariants $\a= -S$ and $\b = - P$, and the same variables were also used in \cite{Kuzenko:2000uh}. Our notation \eqref{SandP} follows \cite{Bandos:2020jsw}. }
\bea
    S 
    = 
    -\frac{1}{4}F^{\mu\nu}F_{\mu\nu}, \quad P = -\frac{1}{4}F_{\mu\nu}\widetilde{F}^{\mu\nu}, \quad \widetilde{F}^{\mu\nu} = \frac{1}{2}\epsilon^{\mu\nu\lambda\tau}F_{\lambda\tau}
    ~.
    \label{SandP}
\eea
It is well-known that only two independent real Lorentz invariant combinations of $F_{\mu\nu}$ can be constructed, and these can be efficiently described by the two quadratic combinations $S$ and $P$ given above.\footnote{For the matrices $F=(F^\m{}_\n) $ and $\widetilde{F}=(\widetilde{F}^\m{}_\n) $, the following identities hold \cite{Schwinger:1951nm}: $F \widetilde{F} = \widetilde{F} F = P {\mathbbm 1} $
and $FF - \widetilde{F}\widetilde{F} = 2 S {\mathbbm 1}$, which allow one to express any invariant of the electromagnetic field in terms of $S$ and $P$. In particular, these identities imply that $F^4 - 2S F^2 - P^2 {\mathbbm 1}=0$ and $(F_\pm)^2 = \hf (S\pm i P) {\mathbbm 1}$, where we have introduced $F_\pm = \hf ( F \pm i \widetilde{F} )$. Therefore, the eigenvalues of $F$ are: $\pm \frac{1}{\sqrt{2}} \big( \sqrt{S+i P} +\sqrt{S-i P} \,\big)$ and $\pm \frac{1}{\sqrt{2}} \big( \sqrt{S+i P} -\sqrt{S-i P}\, \big)$. }
Alternatively, one could use the following two Lorentz invariant combinations of $F_{\mu\nu}$:
\begin{align}
    x_1 = F_{\mu \nu} F^{\nu \mu} = \tr ( F^2 ) \, , \qquad x_2 = F^{\mu \sigma} F_\sigma^{\; \; \, \nu} F_\nu^{\; \; \, \rho} F_{\rho \mu} = \tr ( F^4 ) \, ,
\end{align}
which are related to $S$ and $P$ as
\begin{align}
    x_1 = 4 S \, ,\quad x_2 = 4 P^2 + 8 S^2 ~~~\Longleftrightarrow~~~
    S=\frac{1}{4}x_1  \, ,\quad P=\pm\frac{1}{2}\sqrt{x_2 -\frac{1}{2} (x_1)^2}
    \,.
\end{align}
It is clear that one could use $\cL=\cL(S,P)$ or $\cL=\cL(x_1,x_2)$ as long as one imposes the physical conditions $x_2\geq\hf(x_1)^2$, $x_1\in\mathbb{R}$.

We are interested in families of Lorentz invariant Lagrangians that can be parametrised as $\cL^{(\l)}=\cL(S,P;\lambda)$ or equivalently $\cL^{(\l)}=\cL(x_1,x_2;\lambda)$, with $\lambda$ being, in general, a dimensionful coupling constant and with $\cL^{(\l)}$ being differentiable with respect to $\l$, so that there exists a flow equation
\bea
\frac{\pa\cL^{(\l)}}{\pa\l}:=\cO^{(\l)}
~.
\label{Flow-000}
\eea
Once more, we stress that the operator $\cO$ could be expressed as $\cO^{(\l)}=\cO(S,P;\l)$
or $\cO^{(\l)}=\cO(x_1,x_2;\l)$. 

The equation above can be interpreted geometrically as the statement that the operator $\cO^{(\l)}$ is the tangent vector to a curve in the space of theories, where the points on this curve are the Lagrangians $\cL^{(\l)}$. Given a specific choice of $\cO^{(\l)}$, the same equations can, in principle, be integrated to obtain $\cL^{(\l)}$. This is the same logic used to define models through $\TT$-like flows. These are formally defined as flow equations of the form \eqref{Flow-000} in the special case where the operator is only a function of the energy-momentum tensor $T_{\mu\nu}$, so that $\cO^{(\l)}=\cO(T^{(\l)}_{\mu\nu};\l)$. In fact, a parameterization in terms of the energy-momentum tensor is preferable: it allows us to interpret the tangent vector to the curve as a function only of a particular theory $\mathcal{L}^{(\lambda)}$, rather than depending on the Lorentz invariant kinematic combinations of the electromagnetic field strength in a theory-independent way.

To study classical flow equations, in our paper, we will define $T^{(\l)}_{\mu\nu}$ to be the Hilbert energy-momentum tensor computed from the Lagrangian $\cL^{(\l)}$. A straightforward calculation shows that for a generic Lagrangian $\cL=\cL(F_{\mu\nu})$ the stress tensor is
\bea
T_{\mu \nu} &=& 
\eta_{\mu \nu} \mathcal{L} 
- 4 \frac{\partial \mathcal{L}}{\partial x_1} F^2_{\mu\nu}
- 8 \frac{\partial \mathcal{L}}{\partial x_2} F^4_{\mu\nu}
 \,,
 \label{general_stress_tensor}
 \eea
 with
 \bea
 F^2_{\mu\nu}:=F_{\mu}{}^{\rho}F_{\rho\nu}
 \,,\qquad
 F^4_{\mu\nu}:=F_{\mu}{}^{\rho}F_{\rho}{}^{\tau}F_{\tau}{}^{\sigma}F_{\sigma\nu}
 \,.
\eea
Here, for convenience, we have used in eq.~\eqref{general_stress_tensor} the parametrisation of $\cL$ in terms of $x_1$ and $x_2$, though it is trivial to express the result in terms of $\cL(S,P)$ and its derivatives with respect to  $S$ and $P$ together with the combinations $F^2_{\mu\nu}$ and $F^4_{\mu\nu}$. 

In classifying generic $\TT$-like operators, $\cO(T_{\mu\nu})$, it is useful to identify a basis of Lorentz invariant real scalars obtained from the energy-momentum tensor. For generic Lagrangians $\cL(F_{\mu\nu})$, it suffices to consider the trace of $T_{\mu\nu}$ and the trace of its square:
\bsubeq
\label{traces_T-0}
\bea
\Theta&=&
4\left(\cL
-x_1\frac{\partial \mathcal{L}}{\partial x_1} 
-2x_2 \frac{\partial \mathcal{L}}{\partial x_2} 
\right)
\,,
\label{trace_stress_tensor-0}
\\
T^2
&=&~
16x_2\left( \frac{\partial \mathcal{L}}{\partial x_1} \right)^2 
- 8x_1 \left( x_1^2 - 6 x_2 \right) 
\frac{\partial \mathcal{L}}{\partial x_1}
\frac{\partial \mathcal{L}}{\partial x_2} 
+16\left(x_2^2 + x_1^2 x_2 - \frac{1}{4} x_1^4 \right) 
\left( \frac{\partial \mathcal{L}}{\partial x_2} \right)^2 
\non\\
&&
- 8\mathcal{L} 
\left( x_1 \frac{\partial \mathcal{L}}{\partial x_1} 
+ 2 x_2 \frac{\partial \mathcal{L}}{\partial x_2} \right) 
+4 \mathcal{L}^2 
\, .
\label{trace_T2-0}
\eea
\esubeq
Here, we have introduced the notation
\bea
\Theta:=T^\mu{}_\mu
~,\qquad
T^2:=T^{\mu\nu}T_{\mu\nu}
~.
\eea 
For theories based on a single Abelian gauge field, traces of more than four field strengths $F_{\mu\nu}$ are functions of $x_1$ and $x_2$ only. For this reason, traces of the product of more than two $T_{\mu\nu}$ (e.g. $T_{\mu}{}^\nu T_{\nu}{}^\rho T_{\rho}{}^\mu$) are not independent structures --- see for example the discussion in chapter 7 of \cite{Ferko:2021loo}.   This fact shows that, for this class of Lorentz invariant models, a $\TT$-like flow equation is always going to be of the form
\bea
\frac{\pa\cL^{(\l)}(x_1,x_2)}{\pa\l}=\cO^{(\l)}(T_{\mu\nu})
=\cO^{(\l)}(\Theta,T^2)
=\cO^{(\l)}(x_1,x_2)
~,
\label{flow-001}
\eea
indicating that these flows are always associated with partial differential equations for functions of $x_1$, $x_2$  and of the parameter $\l$ (or of many parameters $\l_i$, $i=1,\cdots,n$, if the Lagrangian has several deformations). 

Equation \eqref{flow-001} could equivalently be expressed as a closed equation in $S$, $P$ and $\l$. In fact, the equations \eqref{traces_T-0} simplify when expressed in terms of $\cL^{(\l)}=\cL(S,P;\l)$. One finds
\bsubeq
\label{traces_T-1}
\bea
\Theta&=&
4\left(
\mathcal{L} 
-P\mathcal{L}_P
-S\mathcal{L}_S
\right)
\,,
\label{trace_stress_tensor-1}
\\
T^2
&=&~
4\left( S^2 + P^2 \right) \mathcal{L}_S^2
+4\left( \mathcal{L} - P \mathcal{L}_P - S \mathcal{L}_S \right)^2
\, ,
\label{trace_T2-1}
\eea
\esubeq
where we have started to use the notation 
$\cL_S:=\frac{\pa \cL}{\pa S}$,
$\cL_P:=\frac{\pa \cL}{\pa P}$,
$\cL_{SP}:=\frac{\pa^2 \cL}{\pa S\pa P}$, etc. 
Interestingly, equation \eqref{trace_T2-1} shows that, for physically relevant models where $\cL_S\ne0$ (such as Maxwell theory and its deformations), $T^2$ is a non-negative number. Moreover, we see that there is a particularly interesting combination given by the trace in Lorentz indices of the square of the traceless part of the energy-momentum tensor:
\bea
\widehat{T}^2
:=
\hat{T}^{\mu\nu}\hat{T}_{\mu\nu}
=4\left( S^2 + P^2 \right) \mathcal{L}_S^2
~,~~~~~~
\hat{T}_{\mu\nu}
=T_{\mu\nu}-\frac{1}{4}\eta_{\mu\nu}\Theta
~,
\label{trace_hatT-0}
\eea
which is also non-negative, $\widehat{T}^2\geq0$. In the following, we will often use $\Theta$ and $\widehat{T}^2$ to parameterise the operator of a  general $\TT$-like deformation $\cO=\cO(T_{\mu\nu};\l)=\cO(\Theta,\widehat{T}^2;\l)$.

Note that the equations \eqref{traces_T-1} define the two Lorentz invariants built from $T_{\mu\nu}$ as functions of $S$ and $P$, so $(\Theta,{T}^2)=(\Theta(S,P),{T}^2(S,P))$. This can be interpreted as a change of variables from $(S, P)$ to $(\Theta, T^2)$. The Jacobian matrix for this transformation is
\begin{align}\label{jacobian}
    J = \begin{bmatrix} \frac{\partial \Theta}{\partial S} & \frac{\partial \Theta}{\partial P} \\ \frac{\partial T^2}{\partial S} & \frac{\partial T^2}{\partial P} \end{bmatrix} \, .
\end{align}
For a generic function $\mathcal{L} ( S, P )$, $J$ is non-degenerate and one can locally invert the change of coordinates as $(S,P)=(S(\Theta,{T}^2),P(\Theta,{T}^2))$. This fact is however misleading since the most interesting physical models (including Maxwell theory, all self-dual models, and all $\TT$-like flows connected to Maxwell) fail to have an invertible map of this type. In fact, it is straightforward to show that, if the Lagrangian $\mathcal{L} ( S, P )$ satisfies the self-duality equation \eqref{SDequation},
\begin{align}
    \mathcal{L}_S^2 - \frac{2 S}{P} \mathcal{L}_S \mathcal{L}_P - \mathcal{L}_P^2 = 1 \, , 
    \label{EM_duality_pde}
\end{align}
then the Jacobian (\ref{jacobian}) for this transformation satisfies
\begin{align}\label{jacobian_zero}
    \det \left( J \right) = 0 \, .
\end{align}
The details of this calculation have been relegated to Appendix \ref{app:determinant}. The vanishing of this Jacobian determinant implies that, in duality-invariant theories, there exists a functional relation of the form
\begin{align}\label{functional_relation_theta_Tsq}
    g \left(  \Theta ( S, P ) , T^2 ( S, P ) \right) = 0
\end{align}
for some function $g$. This means that, locally, one of the functions $\Theta, T^2$ can be written in terms of the other (under mild assumptions on the partial derivatives of the function $g$). 
In the second part of this paper we will see more clearly what form the function $g$ takes for self-dual non-linear electrodynamics formulated in terms of auxiliary fields.

Having introduced various preliminary material, we now focus on understanding how electric-magnetic duality invariance behaves under flows.

\subsection{Duality-invariant theories}
\label{sec:duality_flows}

Electric-magnetic duality in its most basic setting is a symmetry of the equations of motion of free Maxwell theory which is realized as a $\mathbb{Z}_4$ transformation that acts on the field strength and its dual as
\begin{gather}
    F^{\mu\nu}\rightarrow  \widetilde{F}^{\mu\nu}~,\quad \widetilde{F}^{\mu\nu}\rightarrow -F^{\mu\nu}~,\nonumber\\
    \implies F^{\mu\nu}+i( \widetilde{F}^{\mu\nu})\rightarrow e^{i\frac{3\pi}{2}}(F^{\mu\nu}+i( \widetilde{F}^{\mu\nu}))~,
\end{gather}
where the Hodge dual is defined as:
\begin{gather}
    \widetilde{F}^{\mu\nu}= \frac{1}{2}\varepsilon^{\mu\nu\rho\tau}F_{\rho\tau}
    ~.
\end{gather}
This can be elevated to a continuous $U(1)$ transformation, instead of a discrete $\mathbb{Z}_{4}$ action. A theory with Lagrangian $\mathcal{L} ( F_{\mu\nu} )$ is $U(1)$ electric-magnetic duality invariant if the following duality rotation preserves its equations of motion
\begin{align}
\delta_\a F_{\mu \nu} = \a G_{\mu \nu} ( F ) \, ,
\qquad 
\Gt_{\mu \nu} = 2 \frac{\partial \mathcal{L}}{\partial F^{\mu \nu}} 
\,,
\quad
G_{\mu\nu}
=
-\hf\varepsilon_{\mu\nu\rho\tau}\Gt^{\rho\tau}
\,,
\label{duality-rotation-0}
\end{align}
with $\a$ being a real constant parameter. The Lagrangian is generally not invariant under the transformation \eqref{duality-rotation-0}. Once more, a prototypical example is Maxwell's theory with $\cL=S$. 
However, the Euler-Lagrange equations associated with a generic Lorentz invariant Lagrangian $\cL=\cL(S, P)$ respect electric-magnetic duality rotations if equation \eqref{EM_duality_pde} holds \cite{Gaillard:1981rj}.

Given a duality-invariant theory, it is possible to construct large classes of invariant functions which play an important role in our discussion and physically describe observables of self-dual theories. For example, the combination \cite{Gaillard:1997rt,Gaillard:1997zr}
\begin{align}
\mathcal{L} - \frac{1}{4} F \cdot \Gt
~,\qquad
F \cdot \Gt :=F^{\mu\nu}\Gt_{\mu\nu}
\, , 
\end{align}
is duality-invariant.  A short calculation shows that the previous quantity is proportional to the trace of the energy-momentum tensor,
\begin{align}\label{page_one_eqn}
    \mathcal{L} - \frac{1}{4} F \cdot \Gt 
    = \mathcal{L} - S \mathcal{L}_S - P \mathcal{L}_P
    = \frac{1}{4}\Theta \, ,
\end{align}
where the reader should compare with eq.~\eqref{trace_stress_tensor-1}.
In fact, it was proven in \cite{Gaillard:1981rj, Gibbons:1995cv, Gaillard:1997rt,Gaillard:1997zr} that the energy-momentum tensor of a duality-invariant theory is duality invariant, a fact that we will extensively use in the following discussion. This is a simple corollary of the fact that the derivative of $\cL$ with respect to a duality-invariant parameter is duality invariant \cite{Gaillard:1997rt,Gaillard:1997zr}. An instructive example is obtained as follows. If $\cL(F_{\m\n})$ is a solution of the self-duality equation
\eqref{SDequation}, then 
\bea
\cL^{ (g) } (F_{\m\n}) := \frac{1}{g^2} \cL(g F_{\m\n})~, \qquad g \in {\mathbb R}^+ ~, 
\eea
is also a solution of the self-duality equation \eqref{SDequation} in which $\cL$ is replaced with $\cL^{(g)}$ \cite{Kuzenko:2000uh}. According to  \cite{Gaillard:1997rt,Gaillard:1997zr}, the operator $\pa  \cL^{ (g) } / \pa g$ is duality invariant. Direct calculations give 
\bea
\frac{\pa  \cL^{ (g) } }{ \pa g} = -\frac{1}{2 g} \Q^{(g)}~. 
\eea

Now, let $\mathcal{L} ( S,P)$ be the Lagrangian of a $U(1)$ duality-invariant electrodynamics theory. We introduce a one parameter family of deformed theories $\mathcal{L}^{(\l)}(S,P):= \mathcal{L}(S,P; \l )$ defined to satisfy the flow \eqref{Flow-000} with the boundary condition $\mathcal{L}^{(0)}( S,P)=\cL(S,P)$ for some given operator $\cO(S,P;\l)$. We initially do not make further assumptions on $\mathcal{L}^{(\l)}(S,P)$. A natural question to ask is under which conditions the whole family of theories given by $\mathcal{L}^{(\l)}$ is duality invariant if $\mathcal{L}^{(0)}$ is duality invariant.
Remarkably, the following theorem holds:
\begin{theorem}\label{theorem_preserve_invariance} Consider a family of theories satisfying the differential equation and boundary condition
\begin{equation} 
\frac{\pa\cL^{(\l)}(S,P)}{\pa\l}:=\cO^{(\l)}(S,P)=\cO(S,P;\l)
~,\qquad 
\cL^{(0)}(S,P)=\cL(S,P)
~,
\label{Flow-010}
\end{equation}
with $\cO(S,P;\l)$ being a $U(1)$ duality-invariant function, $\d^{(\l)}_\alpha\cO(S,P;\l)=0$.\footnote{The label $\l$ in $\d_\a^{(\l)}$ stresses the fact that the duality transformation \eqref{duality-rotation-0} depends on $\l$.}
If the Lagrangian $\cL(S,P)$ describes a $U(1)$ duality-invariant theory satisfying \eqref{duality-rotation-0}, then all theories associated with the Lagrangians $\cL^{(\l)}(S,P)$ are duality invariant.
\end{theorem}
The first discussion of this property was given in \cite{Ferko:2023ruw}, where it was stated that if $\cO(S,P;\l)=\cO(T^{(\l)}_{\mu\nu};\l)$ then the whole flow of theories is duality invariant --- said differently, $\TT$-like flows preserve $U(1)$-duality invariance. The proof in \cite{Ferko:2023ruw} was sketched, and we provide more detail in our current paper's Appendix \ref{Appendix-A}. Note that, since the stress tensor obeys $\d_\a^{(\l)}T^{(\l)}_{\mu\nu}=0$, any operator $\cO^{(\l)}=\cO(T^{(\l)}_{\mu\nu};\l)$ that is only a function of the energy-momentum tensor evaluated from the Lagrangian $\cL^{(\l)}$ and of the parameter $\l$ is a $U(1)$ duality-invariant function.
Here, we provide an alternative derivation in the case in which $\cO^{(\l)}(S,P)=\cO(S,P;\l)$ is assumed to be a $U(1)$ duality-invariant function and then later we comment on how any invariant function has to be a function of the energy-momentum tensor: $\cO(S,P;\l)=\cO(T^{(\l)}_{\mu\nu};\l)$.

A crucial assumption in the theorem is that $\cO$ is a duality-invariant function. This means that it has to satisfy
\bea
 \mathcal{L}^{(\l)}_S \cO^{(\l)}_S
- \frac{S}{P}\left( \mathcal{L}^{(\l)}_P \cO^{(\l)}_S 
+\mathcal{L}^{(\l)}_S \cO^{(\l)}_P \right)
- \mathcal{L}^{(\l)}_P \cO^{(\l)}_P
=0
~.
\label{invariant-function}
\eea 
This differential equation arises by imposing
\begin{align}
\delta^{(\lambda)}_\a \cO( S , P ; \lambda ) = 0 
\, ,
\qquad {\rm with}\qquad
  \delta^{(\lambda)}_\a F_{\mu \nu} = \alpha G^{(\lambda)}_{\mu \nu} 
\, , \quad 
    \Gt^{(\lambda)}_{\mu \nu} = 2 \frac{\partial \mathcal{L}^{(\lambda)}}{\partial F^{\mu \nu}} \, ,
\end{align}
and explicitly computing
\begin{align}
\delta^{(\l)}_\a \cO^{(\l)} 
= 2 \a \left( P \mathcal{L}^{(\l)}_S \cO^{(\l)}_S
- S \mathcal{L}^{(\l)}_P \cO^{(\l)}_S 
- S \mathcal{L}^{(\l)}_S \cO^{(\l)}_P 
- P \mathcal{L}^{(\l)}_P \cO^{(\l)}_P \right) \, .
\end{align}
Note that eq.~\eqref{invariant-function} was already used in \cite{Ferko:2023ruw} (with $\cO^{(\l)}$ denoted by $f$) to analyse duality-invariance in $\TT$-like flows; see also Appendix \ref{Appendix-A}.

\begin{proof}
Let us now assume that $\mathcal{O}^{(\lambda)}$ is a duality-invariant function and prove the theorem. 
Due to this assumption and eq.~\eqref{Flow-010}, by construction it follows that
\begin{align}
    0 &= \delta^{(\l)}_\a \partial_\l \mathcal{L}^{(\l)} = \partial_\l \delta^{(\l)}_\a \mathcal{L}^{(\l)} 
    = \frac{\a}{2} \partial_\l \left( \Gt^{(\l)} \cdot G^{(\l)} \right)
    = \frac{\a}{2} \partial_\l \left( \Gt^{(\l)} \cdot G^{(\l)} + \Ft \cdot F \right) \, , 
\end{align}
and hence
\begin{align}
    \partial_\l \left( \Gt^{(\l)} \cdot G^{(\l)} + \Ft \cdot F \right) = 0 \, .
\end{align}
The expression in parentheses is an integral of motion for the $\l$ flow. Importantly, its value can be evaluated at $\l = 0$, where it can be shown to be equal to zero. 
As a result,
\begin{align}
    \Gt^{(\l)} \cdot G^{(\l)} + \Ft \cdot F = 0
    \label{EM_duality_pde-2}
\end{align}
along the whole solution of \eqref{Flow-010}.
Let us compute the previous expression explicitly:
\bsubeq
\bea
G^{(\l)}_{\mu\nu}
&=&
\Ft_{\mu\nu}\cL^{(\l)}_S - F_{\mu\nu}\cL^{(\l)}_P
\,,\\
G^{(\l)}\cdot\Gt^{(\l)}
&=&
 4 P \Big{[}
(\cL^{(\l)}_S)^2
-\frac{S}{P}
\cL^{(\l)}_S\cL^{(\l)}_P
-(\cL^{(\l)}_P)^2\Big{]}
\,,
\eea 
\esubeq
%
%
%
and then
\begin{align}
0
=
\Gt^{(\l)} \cdot G^{(\l)} + \Ft \cdot F = 
4 P \Big{[}
(\cL^{(\l)}_S)^2
-\frac{S}{P}
\cL^{(\l)}_S\cL^{(\l)}_P
-(\cL^{(\l)}_P)^2
-1
\Big{]}
\,.
\label{EM_duality_pde-3}
\end{align}
The main point is that \eqref{EM_duality_pde-2} is zero if and only if \eqref{EM_duality_pde} is satisfied for every $\l$. This implies that not only the theory at $\l=0$ is $U(1)$ duality-invariant but the same is true for every $\l$. This concludes the proof of Theorem \ref{theorem_preserve_invariance}.
\end{proof}


Now, we demonstrate that any duality-invariant function $f(S,P)$ in a self-dual theory is a function of the energy-momentum tensor. For this we prove the following theorem:
\begin{theorem}\label{functional_dependence_theorem} 
 Given a $U(1)$ duality-invariant theory with Lagrangian $\cL(S,P)$, any two duality-invariant functions $f(S,P)$ and $g(S,P)$ are functionally dependent.
\end{theorem}
\begin{proof}
First we recall  that $f ( S, P )$ is duality invariant if and only if
\begin{align}\label{f_duality_sergei_notes}
    \left( S \mathcal{L}_P - P \mathcal{L}_S \right) f_S + \left( S \mathcal{L}_S + P \mathcal{L}_P \right) f_P = 0 \, .
\end{align}
To analyse the implications of this condition, we introduce a vector field
on the $(S,P)$-plane,
\begin{align}
    \vec{v} ( S, P ) &= v^S \partial_S + v^P \partial_P
   : = \left( S \mathcal{L}_P - P \mathcal{L}_S \right) \partial_S + \left( S \mathcal{L}_S + P \mathcal{L}_P \right) \partial_P \, .
\end{align}
This vector field is non-vanishing. Otherwise, assuming by way of contradiction that $\vec{v} ( S, P ) = 0$, we would have
\begin{subequations}
\begin{align}
    S \mathcal{L}_P - P \mathcal{L}_S &= 0 \, , \label{contradiction_first} \\
    S \mathcal{L}_S + P \mathcal{L}_P &= 0 \, . \label{contradiction_second}
\end{align}
\end{subequations}
Equation (\ref{contradiction_first}) tells us that
$    \mathcal{L} (S,P)= L \left( S^2 + P^2 \right) $, for some function $L(x)$ of a single variable.
Equation (\ref{contradiction_second}) tells us that $\mathcal{L}(S,P)$ is a homogeneous function of degree $0$, and therefore $\mathcal{L} = {\rm const}$. We have thus arrived at a contradiction.

Equation (\ref{f_duality_sergei_notes}) tells us that the vector field
\begin{align}
    \vec{f} \left( S, P \right) = f_S \partial_S + f_P \partial_P \, , 
\end{align}
is orthogonal\footnote{Here we mean orthogonal with respect to the trivial metric on $\mathbb{R}^2$ with coordinates $(S, P)$, namely $ds^2 = dS^2 + dP^2$. Alternatively, one could say that the one-form $d f = f_S dS + f_P dP$ annihilates the vector $\vec{v}$, $df ( \vec{v} ) = 0 $.} to $\vec{v} ( S, P )$,
\begin{align}
    v^S f_S + v^P f_P = 0 \, .
\end{align}
 Given another duality-invariant function $g ( S, P )$,
\begin{align}
    v^S g_S + v^P g_P = 0 \, , 
\end{align}
both vector fields $\vec{f} ( S, P )$ and
\begin{align}
    \vec{g} ( S, P ) = g_S \partial_S + g_P \partial_P
\end{align}
must be parallel, $\vec{f} \parallel \vec{g}$. This implies
that $\big[f_S , f_P \big] = \l \big[ g_S , g_P\big]$, for some function $\l(S,P)$, and therefore
\begin{align}
    \det \,
    \begin{bmatrix} f_S & f_P \\ g_S & g_P \end{bmatrix} 
    = 0 \, .
\end{align}
This means that the functions $f ( S, P )$ and $g ( S, P )$ are functionally dependent, 
\begin{align}
    \U ( f , g) =0\, ,
\end{align}
 for some function of two variables $\U$.
\end{proof}
Since the energy-momentum tensor $T_{\m\n}$  is duality invariant,  the  duality-invariant functions \eqref{traces_T-1} are functionally dependent, 
equation \eqref{functional_relation_theta_Tsq}. Another corollary of Theorem 
\ref{functional_dependence_theorem} is that any duality-invariant function $f(S,P)$ is a function of the energy-momentum tensor. 
An alternative proof of these results, using the method of characteristics, is given in 
Appendix \ref{app:method_char}.

It is also well-known that any Lagrangian $\mathcal{L} ( S, P )$ which satisfies the duality-invariance condition (\ref{EM_duality_pde}) can also be described in terms of a function of a single independent variable. The logic used to demonstrate this fact is rather different than that reasoning used to establish Theorem \ref{functional_dependence_theorem}, and is also briefly reviewed at the end of Appendix \ref{app:method_char}.

The preceding observations suggest that the analysis of duality-invariant models of electrodynamics, which naively appears to involve functions of two variables $S$ and $P$, can be reduced to a description which involves only functions of a single real variable. To make this intuition and several of these statements more precise, we can employ the auxiliary field formulation of electrodynamics. This will be the focus of the rest of our paper.

\section{Auxiliary field formulation} 
\label{sec:auxreview}

\subsection{Definitions of $\nu$ and $\mu$ representations}
We begin by reviewing the auxiliary field formulation of non-linear electrodynamics first introduced by Ivanov and Zupnik \cite{Ivanov:2001ec}. The two representations used in the rest of this work are the $\nu$ representation and $\mu$ representation.\footnote{We decided to keep using the original nomenclatures of Ivanov-Zupnik but the reader should keep in mind the difference between the variables  $\nu$ and $\mu$ given below and Lorentz indices.} Beginning with the $\nu$ representation, one starts by converting the electromagnetic field strength into spinor notation as follows:\footnote{For the remainder of sections \ref{sec:auxreview} and \ref{sec:stresstensor}, Latin letters represent 4 valued spacetime indices, whilst Greek letters represent 2 valued spinorial indices. See \cite{Buchbinder:1998qv} for our notations and conventions, which mostly agree with those of \cite{Ivanov:2001ec} except, e.g., for the sign of \eqref{FFbar-spin-def}.}
\begin{align}
    \tensor{F}{_\alpha^\beta} 
    = -\frac{1}{4}(\sigma^{\mu})_{\alpha\dot\beta}(\tilde{\sigma}^{\nu})^{\dot\beta\beta}F_{\mu\nu}
    ~,\quad \tensor{\Bar{F}}{_{\dot\alpha}^{\dot\beta}} 
    = \,\frac{1}{4}(\tilde{\sigma}^{\mu})^{\dot\beta\beta}(\sigma^{\nu})_{\beta\dot\alpha}F_{\mu\nu}~,
    \qquad
    F_{\mu\nu}:=2\pa_{[\mu}A_{\nu]}
    ~,
    \label{FFbar-spin-def}
\end{align}
where $(\sigma^{\mu})_{\a\ad},(\tilde{\sigma}^{\mu})^{\ad\a}$ are the Weyl matrices of the group SL$(2,\mathbb{C})$, while $A_\mu$ and $F_{\mu\nu}$ are respectively the gauge connection and field strength of an Abelian gauge theory. One then defines the following Lorentz invariant complex variables:
\begin{align}
    \varphi = F^{\alpha\beta}F_{\alpha\beta},\quad \Bar{\varphi}=\Bar{F}_{\dot\alpha\dot\beta}\Bar{F}^{\dot\alpha\dot\beta}.
\end{align}
With this, one can consider a Lagrangian for non-linear electrodynamics of the form
\begin{align}
    L(\varphi,\Bar{\varphi}) = -\frac{1}{2}(\varphi+\Bar{\varphi})+L^{\text{int}}(\varphi,\Bar{\varphi})
    ~,\label{onshelllagrangian}
\end{align}
with the first monomials describing the Maxwell Lagrangian while $L^{\text{int}}$ is a real function which collects all higher order terms. For instance, one could consider interaction functions $L^{\text{int}}$ which are analytic around $\varphi=0$ and expand in powers $\varphi^{k}\Bar{\varphi}^{m}$, with $k\geq 1$ and $m \geq 1$ (see  \cite{Kuzenko:2000uh} for the technical details). However, we will see that there exist interesting examples of theories for which $L^{\text{int}}$ is non-analytic.

\paragraph{The $\nu$ representation}

With inspiration from the $\mathcal{N}=3$ supersymmetric extension of Born-Infeld theory \cite{Ivanov:2001ec}, Ivanov and Zupnik realised that the kinetic term in equation (\ref{onshelllagrangian}) can be written using an auxiliary, unconstrained antisymmetric real two-form field $V_{\mu\nu}=-V_{\nu\mu}$. We will also define $V_{\alpha \beta} = -\frac{1}{4}(\sigma^{\mu})_{\alpha\dot\beta}(\tilde{\sigma}^{\nu})^{\dot\beta\beta} V_{\mu\nu}$ and  $\bar{V}_{\dot{\alpha} \dot{\beta}} = \frac{1}{4}(\tilde{\sigma}^{\mu})^{\dot\beta\beta}(\sigma^{\nu})_{\beta\dot\alpha} V_{\mu\nu}$, which are the versions of the field $V_{\mu \nu}$ which carry spinor indices, exactly as we have done for the field strength in equation (\ref{FFbar-spin-def}). The result of this rewriting is
\begin{align}
    \mathcal{L}_{2}(V,F) = \frac{1}{2}(\varphi+\Bar{\varphi})+\nu+\Bar{\nu}-2(V\cdot F+\Bar{V}\cdot\Bar{F})
    ~,\label{offshellfree}
\end{align}
where 
\bsubeq 
\bea 
    \nu = V^{\alpha\beta}V_{\alpha\beta}&,&\quad \Bar{\nu} = \Bar{V}_{\dot\alpha\dot\beta}\Bar{V}^{\dot\alpha\dot\beta}
    ~,\nonumber \\
    V\cdot F = V^{\alpha\beta}F_{\alpha\beta}&,&\quad \Bar{V}\cdot\Bar{F} = \Bar{V}_{\dot\alpha\dot\beta}\Bar{F}^{\dot\alpha\dot\beta}
    ~.
\eea 
\esubeq
After integrating out the auxiliary field $V_{\mu\nu}$ via its equation of motion, one arrives back at the free Maxwell Lagrangian. In a straightforward generalisation of the above, a large class of theories of non-linear electrodynamics can be written in the auxiliary field formulation as
\begin{align}
    \mathcal{L}(V,F) = \mathcal{L}_{2}(V,F) + E(\nu,\Bar{\nu})
    ~, \label{offshellfull}
\end{align}
where $E(\nu,\Bar{\nu})$ encodes self-interactions and is all that distinguishes different models. The use of the fields $(F,V)$ as well as relations to come, define the $\nu$ representation. By varying equation (\ref{offshellfull}) with respect to $V_{\alpha\beta}$, one finds the defining algebraic relation between the two fields,
\begin{align}
    F_{\alpha\beta}=V_{\alpha\beta}(1+E_{\nu})~,\quad E_{\nu} = \frac{\partial E(\nu,\Bar{\nu})}{\partial \nu}~.
    \label{EOMV-1}
\end{align}
From this, one finds that the scalar combinations $\varphi$ and $\overbar{\varphi}$ satisfy the relations
\begin{align}
    \varphi = \nu(1+E_{\nu})^2
    ~,\quad F\cdot V = \nu(1+E_{\nu})
    ~,\label{phiinnu}
\end{align}
along with the corresponding complex conjugate relations of \eqref{EOMV-1} and \eqref{phiinnu}.
These equations can, in principle, be solved for $V_{\mu\nu}$ in terms of $F_{\mu\nu}$. In particular, one obtains the following useful relations:
\begin{gather}
    \nu = \varphi G^2~,\quad  V(F)\cdot F = \varphi G
    ~,\label{nuinphi}\nonumber \\
    G = \frac{1}{2}-\frac{\partial L(\varphi,\Bar{\varphi})}{\partial \varphi} = \frac{1}{2}-L_{\varphi}
    ~.
\end{gather}
Using the relations (\ref{phiinnu}), one can transition from a non-linear electrodynamics model to an auxiliary field model via the substitution $(\varphi,\Bar{\varphi})\rightarrow (\varphi(\nu,\Bar{\nu}),\Bar{\varphi}(\nu,\Bar{\nu}))$. Conversely, one can begin with an auxiliary field model $\mathcal{L}(F,V)$ and make the substitution $(\nu,\Bar{\nu})\rightarrow (\nu(\varphi,\Bar{\varphi}),\Bar{\nu}(\varphi,\Bar{\varphi}))$ to recover the non-linear electrodynamics theory formulated only in terms of $F_{\mu\nu}$. This process is outlined further in \cite{Ivanov:2003uj}.

\paragraph{The $\mu$ representation}

The $\mu$ representation is defined via the complex Legendre transform of the $\nu$ frame with the identifications
\begin{gather}
    \mu(\nu,\Bar{\nu}) = E_{\nu}
    ~,\quad \Bar{\mu}(\nu,\Bar{\nu}) = E_{\Bar{\nu}}~,\nonumber\\
    H(\mu,\Bar{\mu}) = E(\nu,\Bar{\nu})-\nu E_{\nu}-\Bar{\nu}E_{\Bar{\nu}}~.\label{legendre}
\end{gather}
The corresponding inverse transformations are
\begin{gather}
    \nu(\mu,\Bar{\mu}) = -H_{\mu}
    ~,\quad \Bar{\nu}(\mu,\Bar{\mu}) = -H_{\Bar{\mu}}
    ~,\nonumber\\
    E(\nu,\Bar{\nu}) = H(\mu,\Bar{\mu})-\mu H_{\mu}-\Bar{\mu}H_{\Bar{\mu}}\label{legendreinverse}
    ~.
\end{gather}
With this, the Lagrangian as well as the defining relation (\ref{phiinnu}) are transformed to
\bsubeq
\begin{gather}
    \mathcal{L}(\varphi,\mu) = \frac{\varphi(\mu-1)}{2(1+\mu)}+\frac{\Bar{\varphi}(\Bar{\mu}-1)}{2(1+\Bar{\mu})}+H(\mu,\Bar{\mu})~,\label{offshellb}\\
    \varphi = -(1+\mu)^2 H_{\mu}~.
    \label{offshellb-2}
\end{gather}
\esubeq
Again, one can recover the non-linear electrodynamics model from (\ref{offshellb}) via the substitution $(\mu,\Bar{\mu})\rightarrow (\mu(\varphi,\Bar{\varphi}),\Bar{\mu}(\varphi,\Bar{\varphi}))$. These are all the essential definitions relevant to the auxiliary field formulation of electrodynamics that we will focus on in our paper. Importantly, we will restrict our attention to the subset of electric-magnetic duality-invariant models. In light of this, we review how electric-magnetic duality acts within this framework, as well as the constraints it imposes.

\subsection{Electric-magnetic duality with auxiliary fields}

In this subsection, we return to the topic of electric-magnetic duality, specifically, the continuous form introduced in Section \ref{sec:duality_flows}.  In spinor notation, the duality rotation of the free Maxwell theory is realised as the infinitesimal transformation
\begin{align}
\delta_{\alpha}
    \begin{pmatrix}
        F_{\alpha\beta}\\
        \Bar{F}_{\dot\alpha\dot\beta}\\
    \end{pmatrix}
    = \begin{pmatrix}
        -i \alpha F_{\alpha\beta}\\
        i\alpha \Bar{F}_{\dot\alpha\dot\beta}\\
    \end{pmatrix}
    ~,
\end{align}
where $\alpha$ is a real parameter (not to be confused with the spinor index $\alpha$). As is well known, the previous transformation is a symmetry of Maxwell equations in the vacuum.

More generally, one can characterize whether a non-linear electrodynamics theory is duality symmetric in the following way. Given a theory with Lagrangian $L(\varphi,\Bar{\varphi})$, the field canonically conjugate to $F_{\alpha\beta}$ is
\begin{align}\label{Galphabeta_defn}
    G_{\alpha\beta}(F) \equiv i\frac{\partial L}{\partial F^{\alpha\beta}} = 2iF_{\alpha\beta}L_{\varphi}
    ~.
\end{align}
This conjugate momentum $G_{\alpha \beta}$ is related to the quantity $\Gt_{\mu \nu}$ of equation (\ref{duality-rotation-0}), although it carries spinor indices rather than Lorentz indices.

The equations of motion and the Bianchi identities for the field $F_{\alpha\beta}$ are given by
\begin{gather}
    \tensor{\partial}{_\alpha^{\dot\beta}}\Bar{G}_{\dot\alpha\dot\beta}-\tensor{\partial}{^\beta_{\dot\alpha}}G_{\alpha\beta} = 0\label{PEOM}
    ~,\\
    \tensor{\partial}{_\alpha^{\dot\beta}}\Bar{F}_{\dot\alpha\dot\beta}-\tensor{\partial}{^\beta_{\dot\alpha}}F_{\alpha\beta} = 0
    ~.\label{offshellBianchi}
\end{gather}
This set of equations is invariant under the transformation
\begin{align}
\delta_{\alpha}
    \begin{pmatrix}
        F_{\alpha\beta}\\
        G_{\alpha\beta}(F)\\
    \end{pmatrix}
    = \begin{pmatrix}
        \alpha G_{\alpha\beta}(F)\\
        -\alpha F_{\alpha\beta}\\
    \end{pmatrix}
    ~,
\end{align}
if the Lagrangian $L(\varphi,\Bar{\varphi})$ satisfies the condition
\begin{align}
    \frac{i}{4}\epsilon^{\mu\nu\rho\tau}(F_{\mu\nu}F_{\rho\tau} + G_{\mu\nu} G_{\rho\tau}) = \varphi-\Bar{\varphi}-4(\varphi(L_{\varphi})^2-\Bar{\varphi}(L_{\Bar{\varphi}})^2) = 0
    ~. \label{offshelldualityconstraint}
\end{align}
Here $G_{\mu \nu}$ is defined by converting $G_{\alpha \beta}$ of equation (\ref{Galphabeta_defn}) to Lorentz indices using the Weyl matrices of $SL(2, \mathbb{C})$. In the $\nu$ representation, the equation of motion for $A_{\alpha\dot\alpha}$ is given by
\begin{gather}
    \tensor{\partial}{^\beta_{\dot\alpha}}(F_{\alpha\beta}(A)-2V_{\alpha\beta}) + \text{c.c.} = 0~,
\end{gather}
and is equivalent to equation (\ref{PEOM}) if one identifies 
\begin{gather}
    G_{\alpha\beta}(F) = i(F_{\alpha\beta}-2V_{\alpha\beta}(F))~.
\end{gather}
Note that one must substitute $V_{\alpha\beta} = V_{\alpha\beta}(F)$ for this to be explicit. With this identification, the realisation of the ${U}(1)$ duality transformations on the independent fields $F_{\alpha\beta}$ and $V_{\alpha\beta}$ is given by
\begin{align}
\delta_{\alpha}
    \begin{pmatrix}
        V_{\alpha\beta}\\
        F_{\alpha\beta}\\
    \end{pmatrix}
    = \begin{pmatrix}
        -i\alpha V_{\alpha\beta}\\
        i\alpha (F_{\alpha\beta}-2V_{\alpha\beta})\\
    \end{pmatrix}
    ~.\label{fvtransformations}
\end{align}
Therefore, by introducing the auxiliary field, a non-linear realisation of $U(1)$ on $(F,G)$ has been transformed into a linear realisation on $(F,V)$. A similar story holds for the $\mu$ representation for the fields $(\varphi,\mu)$. More details can be found in \cite{Ivanov:2003uj}.

The aspect of $U(1)$ duality with which this work is most concerned is the constraints it imposes on the interaction functions $E(\nu,\Bar{\nu})$ and $H(\mu,\Bar{\mu})$. Substituting equation (\ref{PEOM}) into (\ref{offshelldualityconstraint}) and making use of the fact that
\begin{align}
    \nu E_{\nu} = \frac{1}{4}\varphi(1-4L_{\varphi}^2)
    ~,
\end{align}
the  duality condition can be recast as a constraint on the interaction function $E(\nu,\Bar{\nu})$,
\begin{align}
    \nu E_{\nu}-\Bar{\nu}E_{\Bar{\nu}}=0
    ~,
\end{align}
as discussed in \cite{Ivanov:2003uj}.
Under the transformations (\ref{fvtransformations}), the function $E(\nu,\Bar{\nu})$ transforms exactly as above. Hence, the electric-magnetic duality condition can transparently be seen as the requirement that $E(\nu,\Bar{\nu})$ be $U(1)$ invariant. The solution to this constraint is simply a function $\mathcal{E}(a)$ of a single real variable $a = \nu\Bar{\nu}$.
By requiring $\mathcal{E}(a)$ to be analytic and that $\mathcal{E}(0) = 0$, one obtains a smooth weak field limit to Maxwell's Lagrangian. 
The duality invariance is almost identical in the $\mu$ frame as the $U(1)$ invariance of $E(\nu,\Bar{\nu})$ is carried over to $U(1)$ invariance of $H(\mu,\Bar{\mu})$:
\begin{align}
    \delta_{\alpha}H = 2i\alpha(\mu H_{\mu}-\Bar{\mu}H_{\Bar{\mu}}) = 0 \; \Longrightarrow \; H(\mu,\Bar{\mu}) = \cH(b)
    ~,\quad 
    b = \mu\Bar{\mu}
    ~.
\end{align}
Once again, we see that the solution of this constraint is a function $\cH(b)$ in a single real variable $b$.\footnote{Ivanov and Zupnik use $I(b)$ for $H(\mu,\bar{\mu})=\cH(b)$ in the self-dual case.}
With these identifications, the condition for invertibility of the Legendre transform becomes a simple constraint on the derivatives of the interaction functions
\begin{align}\label{Hb_nonzero}
    \mathcal{E}_{a}(0)\neq 0\;\leftrightarrow\; \cH_{b}(0)\neq 0~.
\end{align}
We will return to this constraint in Section \ref{sec:generalargument}. 

The defining relations of the two representations can be simplified using the duality symmetric interaction functions:
\bsubeq\label{useful-FV-1}
\begin{gather}
    \varphi = \nu(1+\Bar{\nu}\mathcal{E}_{a})^2
    ~,\quad \varphi = -(1+\mu)^2\Bar{\mu}\cH_{b}
    ~,\\
    \mathcal{E}(a) = \cH(b)-2b\cH_{b}~,\quad \cH(b) = \mathcal{E}(a)-2a \mathcal{E}_{a}~,
    \\
    \nu (\mu,\Bar{\mu}) = -\Bar{\mu}\cH_{b}~,\quad \mu(\nu,\Bar{\nu}) = \bar{\nu} \mathcal{E}_{a}~.
\end{gather}
\esubeq
It is important to note that the $U(1)$ duality is not a symmetry of the entire auxiliary field Lagrangian. Indeed, the quadratic part transforms as
\begin{align}
    \delta_{\alpha}\mathcal{L}_{2}(F,V)= i\alpha(\varphi-\Bar{\varphi})~.
\end{align}
Therefore, the symmetry holds only for the interaction function and hence it is a ``partial'' symmetry of the entire Lagrangian.

As stated in Section \ref{sec:duality_flows}, it is a well-known fact in the literature that the energy-momentum tensor of a duality-invariant theory will itself always be duality invariant. Hence, if one can show that the energy-momentum tensor is only a function of $\cE(a)$ and vice versa, then it is a very natural question to ask how $\TT$-like deformations of this class of theories behave. This line of reasoning forms the basis of Section \ref{sec:generalargument} and as such, we postpone the rest of the discussion until then.

\subsection{Properties of conformal vs non-conformal models}\label{sec:conf_vs_non_conf}

In order to discuss the dimensionality of various objects in this formalism, we distinguish the cases of conformal and non-conformal models. As will be seen later, the $\mu$ frame is not defined for conformal models as $\cH(b)$ is identically zero and the equations \eqref{offshellb} and \eqref{offshellb-2} are singular. Therefore, it makes sense to treat  the conformal and non-conformal models separately.

One might notice that in the $\mu$ frame, the auxiliary field $\mu$ should be dimensionless in order to not disturb the dimensions of $\varphi$. This would imply that the interaction function $H(\mu,\Bar{\mu})$ is also dimensionless. This is clearly inconsistent as all objects in the Lagrangian must in total have mass dimension $D=4$, such that, overall, the action has units of energy multiplied by time (or be dimensionless in natural units). This means that there is an inherent length scale present in $H(\mu,\Bar{\mu})$ in the form of a dimensionful coupling. Indeed, one can see from the Legendre transform that the dimension of 
$H(\mu,\Bar{\mu})$ is the same as $E(\nu,\Bar{\nu})$. This detail is not present in the $\nu$ representation as the field $\nu$ appears independently in its kinetic part and therefore, has the same units as $\varphi$. Explicit examples of this will be seen in Section \ref{sec:examples}; however, now we discuss the case of conformal models.

\paragraph{Conformal case}

Scale transformations are a subset of conformal transformations, and thus any conformal model must be scale invariant. Due to this, there cannot be any dimensionful parameters present in the Lagrangian. In other words, be it a deformation, or an interaction, all couplings must be marginal. This means that any interaction function ${\cal E}(a)$ of a real variable $a$ for a conformal model must be of the form
\begin{align}
    \mathcal{E}(a) = \mathcal{E}(a;\gamma_{1},...,\gamma_{n}),
\end{align}
where $\{\gamma_{i}\}$ for $i = 1,...,n$ is a set of dimensionless parameters. One well known example of how conformal symmetry can aid this approach is the case of ModMax electrodynamics. Requiring conformal symmetry restricts the interaction function $\mathcal{E}(a)$ to be homogeneous of degree $\frac{1}{2}$ \cite{Kuzenko:2021cvx}. Specifically,
\begin{align}
    \mathcal{E}(a) = \kappa\sqrt{a},
\end{align}
where $\kappa$ is a constant that will need to be determined after integrating out the auxiliary field. For the case of ModMax, one finds 
\begin{align}
    \kappa = 2\tanh\left(\frac{\gamma}{2}\right),
\end{align}
where $\gamma$ is the parameter that moves through the family of theories described by ModMax.

\paragraph{Non-conformal case}

In the case when the model is not conformal, couplings of any dimension are allowed. If a theory has parameters $\{\lambda_{i}\}$ for $i = 1,..,m$ with mass dimension, $[\lambda_{i}]$ with at least one $[\lambda_{j}]\ne0$, one might always choose a single dimensionful coupling and rescale all the others to be dimensionless. The same is true for the $a$ variable. 
Then one can choose to parameterise the interaction functions as follows:
\begin{align}
    \mathcal{E}(a) = \frac{1}{L^4}\mathscr{E}(y;\gamma_{1},...,\gamma_{n-1})
    ~,\qquad 
    [L]=-1~,~~~[\g_i]=0~,~~~
    y=L^{8} a
    ~.
\end{align}
Here ${\sE(y)}$ on the right hand side is  a dimensionless function of $y$ and the couplings $\gamma_i$
which can in principle have arbitrary dependence upon all its variables, in contrast to the conformal case which is is highly constrained.

One can then track how this factor carries through to the definition of $\cH(b)$. From the Legendre transform, one obtains
\begin{align}
    \cH(b) = \frac{1}{L^4}\big(\mathscr{E}(y)-2y\mathscr{E}_{y}\big)
    =\frac{1}{L^4}\mathscr{H}(b)
    ~.
\end{align}
Therefore, in order for this definition not to intrinsically change the Legendre transform, one must also make the change $\cH(b)\rightarrow \cH(b; L^4) = \frac{1}{L^4}\mathscr{H}(b)$. In the case of Born-Infeld and $\gamma$BI, the parameter $L$ is related to the flow parameter that drives the $\TT$-like flow equation. This is special to these two theories.
In order to discuss $\TT$-like deformations of these theories one clearly needs to understand their energy-momentum tensors, a process that we now address.

\section{Energy-Momentum Tensors in $\nu$ and $\mu$ representations} 
\label{sec:stresstensor}
\subsection{Results for duality-invariant theories} 

In order to derive the energy-momentum tensor for duality-invariant auxiliary field models, we work predominantly with the vector form of the Lagrangian and start in the $\nu$ frame:
\begin{align}
    \mathcal{L} = -S-2C-V^{\mu\nu}F_{\mu\nu}+\mathcal{E}(a) \label{vecL}
    ~,
\end{align}
where we define the following scalar combinations of $F$ and $V$:
\begin{gather}
    S = -\frac{1}{4}g^{\mu\nu}g^{\rho\tau}F_{\mu\rho}F_{\nu\tau}
    ~,\quad P = -\frac{1}{4}F_{\mu\nu}\widetilde{F}^{\mu\nu}
    ~,\\
    C = -\frac{1}{4}g^{\mu\nu}g^{\rho\tau}V_{\mu\rho}V_{\nu\tau}
    ~,\quad D = -\frac{1}{4}V_{\mu\nu}\widetilde{V}^{\mu\nu}
    ~.
\end{gather}
Note that we have introduced a generic metric $g_{\mu\nu}$ and its inverse $g^{\mu\nu}$ that will be used to compute the energy-momentum tensor.
In order to perform the variation as well as conversion to spinor components later in the calculation, one will need the following useful relations:
\bsubeq
\begin{gather}
    V^{\mu\nu}V_{\mu\nu} = 2(\nu+\Bar{\nu})
    ~,\quad V_{\mu\nu}\widetilde{V}^{\mu\nu} = 2i(\Bar{\nu}-\nu)
    ~,\\
    \implies \Bar{\nu} = \frac{1}{4}(V^{\mu\nu}V_{\mu\nu}-iV_{\mu\nu}\widetilde{V}^{\mu\nu})
    ~,\quad \nu = \frac{1}{4}(V^{\mu\nu}V_{\mu\nu}+iV_{\mu\nu}\widetilde{V}^{\mu\nu})~,
\end{gather}
\esubeq
with identical relations involving $F_{\mu\nu}$ and $\varphi$.
A generic variation of the Lagrangian (\ref{vecL}) with respect to $g^{\mu\nu}$ is given by
\begin{align}
    \frac{\delta\mathcal{L}}{\delta g^{\mu\nu}} = \frac{1}{2}g^{\rho\tau}F_{\mu\rho}F_{\nu\tau}
    +g^{\rho\tau}V_{\mu\rho}V_{\nu\tau}
    -2g^{\rho\tau}V_{(\mu|\rho|}F_{\nu)\tau}
    +\mathcal{E}_{a}(-g^{\rho\tau}V_{\mu\rho}V_{\nu\tau}C
    +g_{\mu\nu}D^2)
    ~.
\end{align}
It is then straightforward to compute the Hilbert stress-energy tensor via the definition,
\begin{align}
    T_{\mu\nu} &= -\frac{2}{\sqrt{-g}}\frac{\delta(\sqrt{-g}\mathcal{L})}{\delta g^{\mu\nu}}\nonumber\\
    &= g_{\mu\nu}\mathcal{L}
    -g^{\rho\tau}F_{\mu\rho}F_{\nu\tau}
    -2g^{\rho\tau}V_{\mu\rho}V_{\nu\tau}
    +4g^{\rho\tau}V_{(\mu|\rho|}F_{\nu)\tau}
    +2\mathcal{E}_{a}(g^{\rho\tau}V_{\mu\rho}V_{\nu\tau}C
    -g_{\mu\nu}D^2)~.
\end{align}
If one uses the equations \eqref{useful-FV-1}, the trace of the stress tensor is particularly simple,
\begin{gather}
    \Theta:= T_{\mu}{}^{\mu} = 4\mathcal{E}(a)-8a\mathcal{E}_{a}=4\cH(b)
    ~,\label{trace}
\end{gather}
where the last equality is obtained by recalling that the interaction functions $\cH(b)$ and $\mathcal{E}(a)$ are related via a Legendre transform. Note that if the model contains a dimensionful parameter $L$ as per Section \ref{sec:conf_vs_non_conf}, the above relation becomes
\begin{align}
    \Theta = 4\mathcal{H}(b;L^4) = \frac{4}{L^4}\mathscr{H}(b).
\end{align}
This is important for obtaining the correct trace flow equations in Section \ref{sec:examples}. Now, we decompose the vector objects into spinorial components:
\bsubeq
\begin{gather}
    F_{\mu\nu} = \frac{1}{2}\Big(
    (\sigma_{\mu})_{\a\dot\gamma}(\tilde{\sigma}_{\nu})^{\dot\gamma\b}F_{\b}{}^{\a}
    -(\tilde{\sigma}_{\mu})^{\dot\alpha\g}(\sigma_{\nu})_{\g\dot\beta}\Bar{F}^{\dot\beta}{}_{\dot\alpha}
    \Big)
    ~,\\
    V_{\mu\nu} = \frac{1}{2}\Big(
    (\sigma_{\mu})_{\a\dot\gamma}(\tilde{\sigma}_{\nu})^{\dot\gamma\b}V_{\b}{}^{\a}
    -(\tilde{\sigma}_{\mu})^{\dot\alpha\g}(\sigma_{\nu})_{\g\dot\beta}\Bar{V}^{\dot\beta}{}_{\dot\alpha}
    \Big)
    ~.
\end{gather}
\esubeq 
Doing this, and choosing the background metric to be Minkowski ($g_{\mu\nu}=\eta_{\mu\nu}$), the stress tensor for a general electromagnetic duality-invariant auxiliary field model is given by
\begin{align}
    T_{\mu\nu}= \frac{1}{4}\eta_{\mu\nu}\Theta+\hat{T}_{\mu\nu}
    ~,\qquad
    \eta^{\mu\nu}\hat{T}_{\mu\nu}=0
    ~,
\end{align}
with
\begin{equation}
    \hat{T}_{\mu\nu} = \big(1-a(\mathcal{E}_{a})^2\big)(\tilde{\sigma}_{\mu})^{\dot\alpha\alpha}(\tilde{\sigma}_{\nu})^{\dot\beta\beta}V_{\alpha\beta}\Bar{V}_{\dot\alpha\dot\beta}
    = \bigg(\frac{(\mu-1)}{2(1+\mu)}+\frac{(\Bar{\mu}-1)}{2(1+\Bar{\mu})}\bigg)
     (\tilde{\sigma}_{\mu})^{\dot\alpha\alpha}(\tilde{\sigma}_{\nu})^{\dot\beta\beta}F_{\alpha\beta}\Bar{F}_{\dot\alpha\dot\beta}
    ~.
\end{equation}
In the above, we have split the stress tensor into a traceful and traceless part, and we have conveniently used the auxiliary field equations of motion that relate $V_{\alpha \beta}$, $\bar{V}_{\dot{\alpha} \dot{\beta}}$ and $F_{\alpha \beta}$, $\bar{F}_{\dot{\alpha} \dot{\beta}}$ to simplify the expressions. With this, we have the essential building blocks necessary to construct $\TT$-like deformations.

In Section \ref{sec:generalargument} we will argue that there exists a $\TT$-like flow (which is not necessarily unique) for any parameter in an electric-magnetic duality-invariant theory, at least those constructed by using Ivanov-Zupnik's auxiliary field formalisms. However, as discussed in the introduction, two specific stress tensor operators have played a predominant role in previous works since they are associated to Born-Infeld (BI), ModMax, and $\gamma$-BI.
The two such deformations of interest are the usual  four-dimensional $\lambda$ $\TT$ flow and the $\gamma$ $\sqrt{\TT}$ flow.
When discussing such deformations there is always a flow equation associated with the parameters $(\lambda,\gamma)$. Since the auxiliary field Lagrangians are split into a free part ($\mathcal{L}_{2}$) and an interaction function, it is natural to assume that all of the dependence upon the flow parameters sits within the interaction functions only. Explicitly,
\begin{gather}
    \frac{\partial \mathcal{L}(\varphi,\nu)}{\partial \lambda,\gamma}= \frac{\partial \mathcal{E}(a;\lambda,\gamma)}{\partial \lambda,\gamma},\quad \frac{\partial \mathcal{L}(\varphi,\mu)}{\partial \lambda,\gamma}= \frac{\partial \cH(b;\lambda,\gamma)}{\partial \lambda,\gamma}.
\end{gather}
With this, the exact forms of the two deformations in the $\nu$ representation that would appear in equations (\ref{TTbar_flow}) and (\ref{root_TTbar_flow}) are given in the $\nu$ representation by
\begin{align}
    \mathcal{O}_{\TT} &= T^{\mu\nu}T_{\mu\nu}-\frac{1}{2}\Theta^2 = 4a\left(1-a(\mathcal{E}_{a})^2\right)^2-4\left(\mathcal{E}-2a\mathcal{E}_{a}\right)^2
    ~,\\
    \mathcal{O}_{\sqrt{\TT}} 
    &= \frac{1}{2}\sqrt{\Hat{T}^2} = \sqrt{a}\left( 1-a(\mathcal{E}_{a})^2\right)
    ~.
\end{align}
In the $\mu$ representation the operators are given by
\begin{align}
    \mathcal{O}_{\TT} &=-4\mathcal{H}^2+4b(\cH_{b})^2\left(1-b\right)^2
    ~,\\
    \mathcal{O}_{\sqrt{\TT}} &= \sqrt{b}\cH_{b}(1-b)
    ~.
\end{align}
Note that one can obtain the results in both representations by either repeating the process beginning from both auxiliary Lagrangians (\ref{offshellb}, \ref{offshellfull}) separately, or, by completing the process once and converting the results using the relations introduced with the Legendre transformation (\ref{legendre}, \ref{legendreinverse}). Reassuringly, both methods produce the same result. To conclude, it is worth commenting again on the fact that equation (\ref{trace}) implies that the $\mu$ frame has issues with being well defined for conformal models, as in this case $\Theta=0$ and hence $\cH(b)$ is identically zero.

\section{Auxiliary field approach: duality-invariant families as $\TT$-like flows}
\label{sec:generalargument}

So far, we have established several facts about duality-invariant deformations in Section \ref{sec:duality_and_flows}. Let us briefly review three statements.

\begin{enumerate}[label = (\Roman*)]
    \item\label{fact_one} Any deformation of a duality-invariant theory $\mathcal{L}_0$, which is driven by a duality-invariant function $\mathcal{O}^{(\lambda)}$, produces a one-parameter family of duality-invariant theories $\mathcal{L}^{(\lambda)}$ which obey the differential equation $\partial_\lambda \mathcal{L}^{(\lambda)} = \mathcal{O}^{(\lambda)}$. This is Theorem \ref{theorem_preserve_invariance}.

    \item\label{fact_two} The stress-energy tensor of a duality-invariant theory is itself duality-invariant \cite{Gaillard:1981rj, Gibbons:1995cv, Gaillard:1997rt,Gaillard:1997zr}.

    \item\label{fact_three} In a duality-invariant theory, any two duality-invariant functions are functionally dependent. This is Theorem \ref{functional_dependence_theorem}.
\end{enumerate}

Taken together, observations \ref{fact_one} - \ref{fact_three} imply that any deformation of the form $\partial_\lambda \mathcal{L}^{(\lambda)} = \mathcal{O}^{(\lambda)}$ can generically be recast as a stress tensor flow $\partial_\lambda \mathcal{L}^{(\lambda)} = f ( T_{\mu \nu}^{(\lambda)};\l )$, using the functional dependence between the duality-invariant function $\mathcal{O}^{(\lambda)}$ and the energy-momentum tensor, and that the solution to this differential equation is a collection of duality-invariant theories $\mathcal{L}^{(\lambda)}$.

In this section, we will use the $\nu$ and $\mu$ auxiliary field representations to investigate the converse of this result. That is, we ask whether any parameterized family of duality-invariant theories $\mathcal{L}^{(\lambda)}$ can be understood as satisfying some generalized stress tensor flow equation. We will refer to this converse as statement \ref{fact_four}:

\begin{enumerate}[label = (\Roman*)]
\setcounter{enumi}{3}
    \item\label{fact_four} Any family of duality-invariant theories with Lagrangians $\mathcal{L}^{(\lambda)} ( S , P )$ obeys a stress tensor flow equation $\partial_\lambda \mathcal{L}^{(\lambda)} = f \left( T_{\mu \nu}^{(\lambda)};\l \right)$.
\end{enumerate}

There is a simple way to see intuitively why such a converse should be true. We have mentioned above that the derivative of the Lagrangian for a duality-invariant theory, taken with respect to a duality-invariant quantity, is itself duality invariant. Therefore, assuming that the parameter $\lambda$ labeling a family of duality-invariant theories $\mathcal{L}^{(\lambda)}$ does not transform under duality rotations, we must have
\begin{align}\label{converse_pde}
    \frac{\partial \mathcal{L}^{(\lambda)}}{\partial \lambda} = \mathcal{O}^{(\lambda)} ( S , P) 
\end{align}
for some family of functions $\mathcal{O}^{(\lambda)}$, each of which is invariant with respect to the duality transformation associated with the corresponding theory $\mathcal{L}^{(\lambda)}$. Again appealing to observation \ref{fact_three}, we expect that these duality-invariant quantities $\mathcal{O}^{(\lambda)}$ satisfy functional relations involving the respective stress tensors $T_{\mu \nu}^{(\lambda)}$, so that the differential equation (\ref{converse_pde}) can be recast in the form
\begin{align}
    \frac{\partial \mathcal{L}^{(\lambda)}}{\partial \lambda} = f \left(  T_{\mu \nu}^{(\lambda)};\l \right) \, .
\end{align}
Thus we expect that the converse \ref{fact_four} should indeed be true, which leads to a one-to-one correspondence: all stress tensor deformations yield duality-invariant families, and all duality-invariant families are stress tensor flows.

The preceding argument is morally correct. However, to be precise, we should keep in mind that a functional dependence of the form $f ( x, y ) = 0$ only allows us to express $y = y ( x )$ locally around a particular point, and only under the assumption that the appropriate Jacobian determinant is non-zero. In order to give a more careful statement of the converse \ref{fact_four}, we should enumerate the possible singular points at which this Jacobian condition fails, and restrict ourselves to a local analysis away from this collection of singular points.

We can see why this is necessary by considering known examples of stress tensor flows for duality-invariant theories, such as the one for the Born-Infeld and ModMax theories. In terms of the electric and magnetic fields, the Born-Infeld Lagrangian can be written as
\begin{align}
    \mathcal{L}_{\text{BI}} = \frac{1}{\lambda} \left( 1 - \sqrt{ 1 - \lambda \left( \big| \vec{E} \big|^2 - \big| \vec{B} \big|^2 \right) - \lambda^2 \left( \vec{E} \cdot \vec{B} \right)^2 } \right) \, ,
\end{align}
which exhibits a critical value of the electric field,
\begin{align}\label{critical_field_constraint}
    \big| \vec{E} \big|^2 < \frac{1}{\lambda} \, .
\end{align}
Therefore, although we have claimed that the Born-Infeld Lagrangian satisfies a $\TT$-like flow equation
\begin{align}
    \frac{\partial \mathcal{L}_{\text{BI}}}{\partial \lambda} = \frac{1}{8} \left( T^{\mu \nu} T_{\mu \nu} - \frac{1}{2} \left( \tensor{T}{^\mu_\mu} \right)^2 \right) \, , 
\end{align}
to be more precise we should say that this differential equation -- with an initial condition given by the Maxwell Lagrangian -- converges to the Born-Infeld Lagrangian for field configurations within some open set that satisfies the constraint (\ref{critical_field_constraint}). 

A similar caveat applies to the flow for the ModMax theory whose Lagrangian is
\begin{align}
    \mathcal{L}_{\text{ModMax}} = - \frac{1}{4} \cosh ( \gamma ) F_{\mu \nu} F^{\mu \nu} + \frac{1}{4} \sinh ( \gamma ) \sqrt{ \left( F_{\mu \nu} F^{\mu \nu} \right)^2 + \left( F_{\mu \nu} \widetilde{F}^{\mu \nu} \right)^2 } \, .
\end{align}
Clearly the ModMax Lagrangian is not an analytic function of the field strength $F_{\mu \nu}$ and its dual around the point
\begin{align}
    F_{\mu \nu} = \widetilde{F}_{\mu \nu} = 0 \, .
\end{align}
Therefore, when we say that the ModMax Lagrangian obeys a flow equation
\begin{align}
    \frac{\partial \mathcal{L}_{\text{ModMax}}}{\partial \gamma} = \frac{1}{2} \sqrt{ T^{\mu \nu} T_{\mu \nu} - \frac{1}{4} \left( \tensor{T}{^\mu_\mu} \right)^2 } \, , 
\end{align}
we mean that this holds for field configurations away from $F_{\mu \nu} = 0$, on an interval where the Lagrangian is an analytic function. This example illustrates that flow equations may hold at generic points but fail at a discrete collection of exceptional points where singularities occur.

However, if we are willing to restrict ourselves to a local analysis motivated by the examples of $\mathcal{L}_{\text{BI}}$ and $\mathcal{L}_{\text{ModMax}}$ described above, the desired converse \ref{fact_four} holds away from a set of discrete singular points which we can identify explicitly. We will state the precise version of this claim for theories that admit a $\nu$ representation. As explained in \cite{Ivanov:2003uj}, any Lagrangian $\mathcal{L} ( S, P )$ satisfying the duality-invariant condition can be described in the $\nu$ frame, so there is no loss of generality in using this representation.

\begin{theorem}\label{nu_thm} Consider a family of theories of duality-invariant electrodynamics which are labeled by a collection of parameters $\l_i$, $i = 1 , \cdots , n$, and which admit a description using the $\nu$ representation introduced in equation (\ref{offshellfull}). That is, the entire parameterized family of Lagrangians is determined by an interaction function
\begin{align}
    \mathcal{E} ( a;\l_1 , \ldots , \l_n ) \, , 
\end{align}
where $a = \nu \overbar{\nu}$. Let $a^\ast \in \mathrm{dom} ( \mathcal{E} ) \setminus S$, that is, let $a^\ast$ be a point in the domain of $\mathcal{E}$ which does not belong to the discrete (possibly empty) set of points $S$ defined by
\begin{align}\label{discrete_set_S}
    S = &\Bigg\{ \underline{a} \; \, \Big| \, \; \left( \mathcal{E}_a ( \underline{a} ) = \frac{1}{\underline{a}} \, \text{ and } \, \mathcal{E}_{aa} ( \underline{a} ) = - \frac{1}{2 \underline{a}^2} \; \right) \nonumber \\
    &\qquad \qquad \text{ or } \; \left( \mathcal{E}_a ( \underline{a} )  = \frac{1}{2 \underline{a}} \, \text{ and } \, \mathcal{E}_{aa} ( \underline{a} ) = - \frac{1}{4 \underline{a}^2} \right) \Bigg\} \, .
\end{align}
Then for each $i$ there exists an open interval $U_i$ around $a^\ast$ such that, on the set $U_i$, one has
\begin{align}
    \frac{\partial \mathcal{E}}{\partial \l_i} = F_i \left(T_{\mu \nu} ; \l_1 , \ldots , \l_n \right) \, , 
\end{align}
where $F_i$ is a Lorentz scalar constructed from the stress tensor $T_{\mu \nu}$ and which may depend on the parameters $\l_i$.
\end{theorem}

The interpretation of this theorem is that, at least locally, every tangent vector to a space of theories of duality-invariant electrodynamics is an operator constructed from the stress tensor. This result is the precise version of statement \ref{fact_four}, the desired converse to the result that stress tensor deformations preserve duality invariance.

\begin{proof}

The proof of this claim is a simple application of the inverse function theorem. We begin by parameterizing the most general Lorentz scalar function which can be constructed from the energy-momentum tensor $T_{\mu \nu}$. A convenient basis for the ring of scalars that can be built from the stress tensor of a duality-invariant theory in the $\nu$ frame is
\begin{align}\label{two_stress_tensor_invariants}
    \Theta = 4 \mathcal{E} - 8a \mathcal{E}_a \, , \qquad \widehat{T}^2 = 4 a \left( 1 - a \mathcal{E}_a \right)^2 \, .
\end{align}
Any other Lorentz scalar built from $T_{\mu \nu}$ can be written as a function of $\Theta$ and $\widehat{T}^2$. Although $\Theta$ and $\widehat{T}^2$ implicitly depend on the $\l_i$ that determine $\mathcal{E}$, let us hold these parameters fixed for the moment and consider the invariants (\ref{two_stress_tensor_invariants}) as univariate functions of the real variable $a$. The derivatives of these functions are
\begin{align}\label{two_derivatives}
    \frac{d \Theta}{da} &= -4 \left( \mathcal{E}_a + 2 a \mathcal{E}_{aa} \right) \, , \nonumber \\
    \frac{d}{da} \left( \widehat{T}^2 \right) &= 4 \left( 1 - a \mathcal{E}_a \right) \left( 1 - 3 a \mathcal{E}_a - 2 a^2 \mathcal{E}_{aa} \right) \, .
\end{align}
Let us consider the conditions under which both of the derivatives in (\ref{two_derivatives}) can vanish simultaneously at a point $a = \underline{a}$. First, there are two ways for $\frac{d \Theta}{da}$ to vanish at $\underline{a}$:
\begin{align}\label{deriv_cases}
    \frac{d \Theta}{da} \Big\vert_{a = \underline{a}} = 0 \; \implies \; \begin{cases}
\underline{a} = 0  \, , \quad \mathcal{E}_a ( \underline{a} ) = 0 \, \\
\underline{a} \neq 0 \, , \quad \mathcal{E}_{aa} = - \frac{\mathcal{E}_a}{2a} 
\end{cases} \, .
\end{align}
If $\underline{a} = 0$ and $\mathcal{E}_a ( \underline{a} ) = 0$, we have $\frac{d \widehat{T}^2}{da} = 4$, so this is not a point at which both derivatives can vanish. Therefore suppose that we are in the second case of (\ref{deriv_cases}). We then have
\begin{align}
    \left[ \frac{d \widehat{T}^2}{d a} \right]_{\mathcal{E}_{aa} = - \frac{\mathcal{E}_a}{2a}} = 4 \left( 1 - a \mathcal{E}_a  \right) \left( 1 - 2 a \mathcal{E}_a \right) \, , 
\end{align}
which means either
\begin{align}\label{zero_deriv_conditions}
    \mathcal{E}_a ( \underline{a} ) = \frac{1}{\underline{a}} \, , \, \mathcal{E}_{aa} ( \underline{a} ) = - \frac{1}{2 \underline{a}^2} \; \text{ or } \; \mathcal{E}_a ( \underline{a} )  = \frac{1}{2\underline{a}} \, , \, \mathcal{E}_{aa} ( \underline{a} ) = - \frac{1}{4 \underline{a}^2} \, .
\end{align}
It is not possible for either of the pairs of conditions (\ref{zero_deriv_conditions}) to hold on an open set. If $\mathcal{E}_a = \frac{1}{a}$ on some open set $U$, then $\mathcal{E}_{aa} = - \frac{1}{a^2}$ on this set, which does not agree with the condition $\mathcal{E}_{aa} = - \frac{1}{2 a^2}$. Likewise, if $\mathcal{E}_a = \frac{1}{2a}$ for all $a \in U$, then $\mathcal{E}_{aa} = - \frac{1}{2 a^2}$ within $U$, which disagrees with the condition $\mathcal{E}_{aa} = - \frac{1}{4 a^2}$. Therefore, either pair of conditions (\ref{zero_deriv_conditions}) can hold only on a discrete set of points, and at any other point off this set we have that either $\frac{d \Theta}{da} \neq 0$ or $\frac{d \widehat{T}^2}{da} \neq 0$.

Therefore, given any value $a^\ast$ of the auxiliary variable which is on the complement of a discrete set $S$ of exceptional points introduced in equation (\ref{discrete_set_S}), we can define a function $f(a)$ which is a Lorentz scalar constructed from the stress tensor and such that $f' ( a^\ast ) \neq 0$. Explicitly, let $f(a) = \Theta ( a )$ if $\frac{d \Theta}{da} \Big\vert_{a^\ast} \neq 0$ and let $f(a) = \widehat{T}^2 ( a )$ if $\frac{d \widehat{T}^2}{da} \Big\vert_{a^\ast} \neq 0 $; if both derivatives are non-zero, we can choose $f$ at random (in this case, since $\Theta$ and $\widehat{T}$ can be expressed in terms of one another by Theorem \ref{functional_dependence_theorem}, these choices are equivalent).

By the inverse function theorem, there exists an open interval $U$ containing $a^\ast$ such that the restriction of the function $f$ to $U$ is a bijection and there exists a differentiable inverse function $f^{-1}$. This means that, locally, the variable $a$ can be written as a differentiable function of the variable $f$, which again is either $\Theta$ or $\widehat{T}^2$. By composing with $f^{-1}$, we conclude that any function of $a$ can be written as a function of the stress tensor in a neighborhood of $a^\ast$. In particular, this conclusion applies to each of the functions
\begin{align}\label{functions_conclusion}
    \frac{\partial \mathcal{E}}{\partial \l_i} \, ,
\end{align}
viewed as univariate functions of $a$ with the parameters $\l_i$ held fixed. This demonstrates that each of the derivatives (\ref{functions_conclusion}) can locally be written as a function of $\Theta$ or $\widehat{T}^2$, along with the parameters $\l_i$, which establishes the claim.
\end{proof}

As a simple example of one of the discrete points $\underline{a} \in S$ at which the claim can fail, consider the interaction function
\begin{align}
    \cE ( a ) = \l_0+\l_1a+\lambda_2 a^2  \, .
    \label{simple-000}
\end{align}
Other quadratic examples of interaction functions will be explored in Section \ref{simple_instructive}. At the point
\begin{align}
    \lambda_1 = \frac{3}{2} \, , \quad  \lambda_2 = -\frac{1}{4} \, , \quad a = 1 \, , 
\end{align}
we find
\begin{align}
    \Theta = - \frac{3}{2} + 4 \lambda_0 \, , \quad \widehat{T}^2 = 1 \, , \quad \frac{d \Theta}{da} = 0 \, , \quad \frac{d \widehat{T}^2}{da} = 0 \, .
\end{align}
Near this point, it is not possible to locally express $a$ as a function of $\Theta$ and $\widehat{T}^2$, and thus we cannot write flow equations of the form $\frac{d \mathcal{E}}{\partial \l_i} = f ( \Theta , \widehat{T}^2 )$. However, because such exceptional points are isolated, given any $\epsilon > 0$, there exists some $a'$ within distance $\epsilon$ of $a$, with the property that we may express these derivatives in terms of the stress tensor near the point $a'$. Said differently, the tangent vector to a family of duality-invariant theories can be written as a function of the stress tensor ``almost everywhere'' (that is, at all points $a$ except on a set of measure zero). We will elaborate more on this simple example in Section \ref{simple_instructive}.

\subsubsection*{\ul{\it Proof in $\mu$ frame}}

One can present an analogous argument, showing that families of duality-invariant theories can generically be interpreted as stress tensor flows, using the other auxiliary field representation, which was referred to as the $\mu$ frame in Section \ref{sec:auxreview}.

This proof is less general because not all duality-invariant theories admit a description in the $\mu$ representation. For instance, we have pointed out above that conformal models such as the Maxwell and ModMax theories cannot be written in the $\mu$ frame.

It is easy to see why there is an obstruction to describing certain models in the $\mu$ representation if we recall the relationship between quantities in the $\mu$ and $\nu$ frames. Consider a duality-invariant theory of electrodynamics which is described by an interaction function $\cE ( a )$, where $a = \nu \overbar{\nu}$, in the $\nu$ representation. The corresponding interaction function $\cH( b )$, where $b = \mu \overbar{\mu}$, in the $\mu$ representation satisfies the relation
\begin{align}\label{no_vanishing_Ib}
    \mathcal{E}_a = - \frac{1}{\mathcal{H}_b} \, .
\end{align}
This equation admits solutions only if $\mathcal{E}_a \neq 0$ and $\mathcal{H}_b \neq 0$; we have already seen this additional condition on $\mathcal{H}$ in equation (\ref{Hb_nonzero}). On the other hand, the two invariants constructed from the stress tensor in the $\mu$ frame take the forms
\begin{align}
    \Theta = 4 \mathcal{H} ( b ) \, , \qquad \widehat{T}^2= 4 b \mathcal{H}_b^2 \left( 1 - b \right)^2 \, .
\end{align}
For a conformal model, $\Theta = 0$ which means that $\mathcal{H}(b)$ is identically zero. But if $\mathcal{H}(b) = 0$, then $\mathcal{H}_b = 0$ and the relation (\ref{no_vanishing_Ib}) is not well-defined. We conclude that the $\mu$ representation is only suitable for describing theories with $\mathcal{H}_b \neq 0$ and thus $\Theta_b \neq 0$, which excludes conformal models with $\Theta = 0$.

This allows us to give a very simple proof of the analogue of Theorem \ref{nu_thm} for theories with a $\mu$-frame representation.

\begin{corollary}\label{mu_thm}

Consider a family of theories of duality-invariant electrodynamics which are labeled by a collection of parameters $\l_i$, $i = 1 , \cdots , n$, and which admit a well-defined description using the $\mu$ representation introduced in equation (\ref{offshellb-2}). That is, the entire parameterized family of Lagrangians is determined by an interaction function
\begin{align}
    \mathcal{H} (b; \l_1 , \ldots , \lambda_n ) \, , 
\end{align}
where $b = \mu \overbar{\mu}$. Then for each $i$ any for any point $b = b^\ast$, there exists an open interval $U_i$ around $b^\ast$ such that, on the set $U_i$, one has
\begin{align}\label{I_deriv}
    \frac{\partial \mathcal{H}}{\partial \l_i} = F_i \left(\Theta; \l_1 , \ldots , \l_n  \right) \, , 
\end{align}
where $F_i$ is a Lorentz scalar constructed from the trace $\Theta$ of the stress tensor and which may depend on the parameters $\gamma_i$.
    
\end{corollary}

We point out that Corollary \ref{mu_thm} differs from the corresponding $\nu$-frame statement, Theorem \ref{nu_thm}, in two ways. First, in the $\mu$ frame we need not make any additional assumption about being away from a discrete set of exceptional points. Second, in the $\mu$ frame we can always express the deforming operators which drive the flows in equation (\ref{I_deriv}) in terms of only the trace of the stress tensor, rather than as a function of the two invariants $\Theta$ and $\widehat{T}^2$. Of course, as we have emphasized, these two scalars are functionally dependent in any duality-invariant theory; the only new feature in the $\mu$ frame is that cases with $\partial_b \Theta = 0$ are excluded. 

\begin{proof}

By assumption, the family of theories that we are considering admit a well-defined $\mu$-frame description, which implies that $\mathcal{H}_b \neq 0$ as we pointed out around equation (\ref{no_vanishing_Ib}). On the other hand, we have the relation
\begin{align}
    \Theta = 4 \mathcal{H} ( b ) \, .
\end{align}
Since $\mathcal{H}_b \neq 0$, we also have $\frac{d \Theta}{db} \neq 0$, and thus by the inverse function theorem we may always locally invert this relation to write $b = b ( \Theta )$ on a sufficiently small open set $U$ around any particular point $b = b^\ast$.

It follows that any function of $b$ can locally be expressed as a function of $\Theta$, and thus
\begin{align}
    \frac{\partial \mathcal{H}}{\partial \gamma_i} &= F_i \left(b; \l_1 , \cdots , \l_n  \right) \nonumber \\
    &= F_i (\Theta ; \l_1 , \cdots , \l_n ) 
\end{align}
on an open set $U_i$ containing any point $b^\ast$.
\end{proof}

\section{Examples} 
\label{sec:examples}

We have seen that there is a one-to-one correspondence between parameterized families of duality-invariant theories and stress tensor flows. This correspondence is summarized in the facts \ref{fact_one} - \ref{fact_four} of the previous section and in the statements of the various theorems where these results are proved.

In one sense, this means that the primary task of the present work has been completed. However, we find it instructive to examine several examples where this one-to-one correspondence can be described explicitly in both directions. It is especially interesting to see how one can determine the stress tensor operator which forms the tangent vector to a given family of duality invariant theories, which gives a concrete realization of statement \ref{fact_four}.

In the following subsections, we will carry out this procedure in several examples using the $\nu$ and $\mu$ frame auxiliary field representations. This will allow us to build further intuition for the singular points, belonging to the set $S$ of equation (\ref{discrete_set_S}), where the inversion map between the duality-preserving deformation and the energy-momentum tensor breaks down. We will see that such points often arise from non-analyticity in the interaction function.

We will also revisit the flow equations which produce the ModMax and Born-Infeld theories from stress tensor flows from the perspective of the auxiliary field formalism. As an extension of this analysis, in Section \ref{subsec:gamma_BI}, we will obtain a new auxiliary field representation of the two-parameter family of ModMax-Born-Infeld theories, which we also call $\gamma$BI.

\subsection{Quadratic interaction functions}\label{simple_instructive}

We begin with the simplest two-parameter family of models described by
\begin{align}
    \mathcal{E} (a; \l_1, \l_2 ) = \l_1 a + \l_2 a^2 \, .
\end{align}
We have already discussed this model around eq.~\eqref{simple-000}, though, for simplicity, we set here $\l_0=0$.
The trace of the stress tensor is
\begin{align}
    \Theta = - 8 a \left( \l_1 + 2 a \l_2 \right) + 4 \left( \l_1 a + \l_2 a^2 \right) \, .
\end{align}
Treating the two $\l_i$ as constants for the moment, we can view this as a simple univariate function $\Theta ( a )$ and utilise the inverse function theorem. In this case, there is a single point at which the assumption of the inverse function theorem fails because $\Theta ( a )$ has zero derivative:
\begin{align}\label{x_star}
    \frac{d \Theta}{da} = 0 \; \text{ at } \; a^\ast = - \frac{\l_1}{6 \l_2} \, .
\end{align}
Away from this point, one can simply solve to express $a$ as a function of $\Theta$, finding
\begin{align}\label{x_root_example}
a = \frac{- \l_1 \pm \sqrt{ \l_1^2 - 3 \Theta \l_2} }{6 \l_2} \, .
\end{align}
From this, it is easy to see why the point (\ref{x_star}) is problematic. This value of $a$ occurs when the argument of the square root vanishes, $\l_1^2 = 3 \Theta \l_2$, and the square root function is not analytic around zero. If we assume that $\l_1^2 > 3 \Theta \l_2$, and choose the positive root of (\ref{x_root_example}) so that $a > 0$ (which is expected since $a = \nu \overbar{\nu}$), then there are no such issues, and we can write
\begin{align}
    \partial_{\l_1} \mathcal{E} = \frac{\sqrt{\l_1}}{3 \l_2} \left( - \l_1 + \sqrt{ \l_1^2 - 3 \Theta \l_2} \right) \, .
\end{align}
Note that the second stress tensor invariant $\widehat{T}^2$ was not needed at all for this procedure. However, we could have made a similar argument as above, viewing $\widehat{T}^2$ as a function of $a$ and using the inverse function theorem again. In this case,
\begin{align}
    \widehat{T}^2 ( a ) = 4 a \left( 1 - a \left( \l_1 + 2 a \l_2 \right) \right)^2 \, ,
\end{align}
and the inverse function theorem fails due to a vanishing derivative $\frac{d}{dx} \widehat{T}^2$ at four points:
\begin{align}\label{four_roots}
    a_{1, 2} = \frac{- \l_1 \pm \sqrt{ \l_1^2 + 8 \l_2}}{4 \l_2} \, , \quad a_3 = \frac{- 3 \l_1 - \sqrt{ 9 \l_1^2 + 40 \l_2 }}{20 \l_2} \, , \quad a_4 = \frac{2}{3 \l_1 + \sqrt{ 9 \l_1^2 + 40 \l_2}} \, .
\end{align}
Away from these four points $a_1, a_2, a_3, a_4$, we see that $\widehat{T}^2$ is a smooth function of $a$ and we are guaranteed by the inverse function theorem that we can invert to write $a ( \widehat{T}^2 )$ on an open interval that does not include any of the roots (\ref{four_roots}). We cannot write this inverse function explicitly because it involves the root of a fifth-order polynomial, but it is sufficient to know that it exists in order to claim that
\begin{align}
    \partial_{\l_1} \mathcal{E} = F \left( \widehat{T}^2;\l_1 , \l_2  \right)
\end{align}
for some function $F$. Nothing was special about choosing $\l_1$ in the above analysis. If we had focused on $\l_2$, we could invert for $a$ in the same way. For instance, one has
\begin{align}
    \partial_{\l_2} \mathcal{E} = \left( - \l_1 + \sqrt{ \l_1^2 - 3 \Theta \l_2} \right)^2 \, .
\end{align}
The above arguments then generalise as we described in Section \ref{sec:generalargument}.

The situation is even simpler in the $\mu$ representation. The simplest model in this case is given by:
\begin{align}
    \mathcal{H}(b;L^4,\gamma_{1},\gamma_{2}) = \frac{1}{L^4}(\gamma_{1} b+\gamma_{2}  b^2)
    ~.
\end{align}
Here $\gamma_1$, $\gamma_2$ are arbitrary constants which have no relation to the $\lambda_1$, $\lambda_2$ of the previous example. In the $\mu$ representation, we always have the trace relation
\begin{align}
    \Theta = 4\mathcal{H}(b;L^4,\gamma_{1},\gamma_{2}).
\end{align}
As stated in Section \ref{sec:generalargument}, the main conditions for using the $\mu$ representation are that $\mathcal{H}_{b}\neq 0$ and $\Theta_{b}\neq 0$. We can look for the points at which this assumption fails by computing
\begin{align}
    \Theta_{b} = \frac{4}{L^4}(\gamma_{1}+2\gamma_{2} b) \, .
\end{align}
One can straightforwardly see that there is only one point ($b = b^{*}$) at which the derivative vanishes,
\begin{align}
    b^{*} = -\frac{\gamma_{1}}{2\gamma_{2}}.
\end{align}
Again, the meaning of this point becomes clear by solving for $b(\Theta)$:
\begin{align}
    b = \frac{-\gamma_{1}\pm\sqrt{\gamma_{1}^2+L^4\gamma_{2}\Theta}}{2\gamma_{2}}.
\end{align}
The issue of invertability is linked to the vanishing of the expression inside the square root. If we assume that $\gamma_{1}^2 > -L^4\gamma_{2}\Theta$, then the issue is avoided and the flow equations for the parameters $(L^4, \gamma_{1},\gamma_{2})$ are
\bsubeq
\bea
    \frac{\partial \mathcal{H}(b;L^4,\gamma_{1},\gamma_{2})}{\partial L^4} &=& -\frac{1}{4L^4}\Theta
    ~,\\
    \frac{\partial \mathcal{H}(b;L^4,\gamma_{1},\gamma_{2})}{\partial \gamma_{1}} &=& \frac{1}{L^4}\left(\frac{-\gamma_{1}\pm\sqrt{\gamma_{1}^2+L^4\gamma_{2}\Theta}}{2\gamma_{2}}\right)
    ~,
\\ 
    \frac{\partial \mathcal{H}(b;L^4,\gamma_{1},\gamma_{2})}{\partial \gamma_{2}} &=& \frac{1}{L^4}\left(\frac{-\gamma_{1}\pm\sqrt{\gamma_{1}^2+L^4\gamma_{2}\Theta}}{2\gamma_{2}}\right)^2
    ~.
\eea
\esubeq
Just as in the $\nu$ frame, we could have inverted the expression for $\Hat{T}^2$ instead of $\Theta$, given by
\begin{gather}
    \Hat{T}^2 = b\cH_{b}^2(1-b)^2 = \frac{1}{L^8}(\gamma_{1}+2\gamma_{2}b)^2(1-b)^2 \, , \label{that3}\\
    \rightarrow b = \frac{-\gamma_{1}+2\gamma_{2}\pm\sqrt{(\gamma_{1}+\gamma_{2})^2\pm8\gamma_{2}L^4\sqrt{\Hat{T}^2}}}{4\gamma_{2}} \, .
\end{gather}
The points where the derivative of (\ref{that3}) vanishes are
\begin{gather}
    b^{*} = 1, \quad b^{*} = -\frac{\gamma_{1}}{2\gamma_{2}},\quad b^{*} = \frac{1}{2}-\frac{\gamma_{1}}{4\gamma_{2}} \, .
\end{gather}
With this, we can now look at examples that are more complex than polynomial interactions. The first point appears due to the $(1-b)^2$ factor in (\ref{that3}), whilst the later two points are linked to the vanishing of the square root. Once again, if this does not occur, then we can at least locally write the inverse $b(\Hat{T}^2)$ away from these points.

\subsection{Born-Infeld and ModMax}

One of the most well studied examples of this formalism is Born-Infeld theory. The formulation of this theory with one auxiliary field was first introduced by Ivanov and Zupnik \cite{Ivanov:2001ec} and is the model which we will start by reviewing. In four spacetime dimensions, the Born-Infeld Lagrangian is given by
\begin{align}
    L_{\rm BI}(\varphi,\Bar{\varphi}) = \frac{1}{\lambda}\left(\,
    1-\sqrt{1+\lambda(\varphi+\Bar{\varphi})+\frac{\lambda^2}{4}(\varphi-\Bar{\varphi})^2}\,\right).
\end{align}
This theory is best studied in the $\mu$ frame, due to the simplicity of the resulting interaction function $\mathcal{H}(b)$. The defining relations for Born-Infeld  in the $\mu$ representation are
\bsubeq
\begin{gather}
    \varphi = \frac{2\Bar{\mu}(1+\mu)^2}{\lambda(1-\mu\Bar{\mu})^2}~,\quad \Bar{\varphi} = \frac{2\mu(1+\Bar{\mu})^2}{\lambda(1-\mu\Bar{\mu})^2}
    ~,\\
    \mathcal{H}_{b} = -\frac{2}{\lambda(b-1)^2}\:\rightarrow\: \mathcal{H}(b,\lambda) = \frac{1}{\lambda}\frac{2b}{b-1}~.
\end{gather}
\esubeq 
In the above, one notices the factorisation mentioned in Section \ref{sec:conf_vs_non_conf} for a non-conformal model with $\l=L^4$. Using the previous expressions for the interaction function, one obtains the following auxiliary field Lagrangian for Born-Infeld:
\begin{align}
    \mathcal{L}(\varphi,\mu) = \frac{\varphi(\mu-1)}{2(1+\mu)}+\frac{\Bar{\varphi}(\Bar{\mu}-1)}{2(1+\Bar{\mu})}+\frac{1}{\lambda}\frac{2b}{b-1}
    ~.\label{offshellBI}
\end{align}
Using the results of Section \ref{sec:stresstensor}, the trace of the stress tensor is given by
\begin{gather}
    \Theta(b) = \frac{4}{\lambda^2}\frac{2b}{b-1} \; \longrightarrow \; b(\Theta) = -\frac{\Theta\lambda^2}{8-\Theta\lambda^2}
    ~.
\end{gather}
The flow equation relating to $\lambda$ is simply the $\TT$ flow equation,
\begin{align}
    \frac{\partial \mathcal{L}(\varphi,\mu)}{\partial \lambda} =\frac{1}{8}\mathcal{O}_{\TT} = -\frac{1}{4\lambda}\Theta~.\label{BIflow}
\end{align}
From this example, we can see explicitly that the operator associated to the flow is not unique. However, we stress that only the flow driven by the $\mathcal{O}_{\TT}$ operator
can be used if one would like to interpret the Born-Infeld Lagrangian as a $\TT$-like flow with a boundary condition at $\l=0$ being free Maxwell. The trace flow is indeterminate in this limit, since both the numerator $\Theta$ and denominator $\lambda$ of the right side of (\ref{BIflow}) vanish.  
As stated in Section \ref{sec:generalargument}, the trace flow equation is something we will always have for the dimensionful parameter due to the relationship between $\cH(b)$ and $\Theta$.

At this point, one can pass to the $\nu$ representation by solving for $b(a)$ through the following algebraic relation, which can be derived using equation (\ref{legendreinverse}),
\begin{gather}
    a = b\mathcal{H}_{b}^2= \frac{4b}{\lambda^2(b-1)^4}
    ~.
\end{gather}
Introducing $t = (b-1)^{-1}$, then one can find a closed form expression for $t(a)$ which solves the following quartic equation:
\bea
    t^4+t^3-\frac{\lambda^2}{4}a = 0
    ~,\quad 
    t(a=0) = -1
    ~.
\eea
The solution $t(a)$ is fairly involved, and, for brevity, we present the first terms in its power series:
\bea
    t(a) = -1-\frac{\lambda^2 a}{4}+\frac{3\lambda^4a^2}{16}
    +\cdots
    ~.
\eea
Finally, one can use the Legendre transform to find the interaction function in the $\nu$ representation, for which, due to the nature of $t(a)$, we  also present only the first few terms in its series expansion:
\begin{align}
    \mathscr{E}_{BI}(y) = 2\left(t^2(a)+3t(a)+1\right) = \frac{y}{2}-\frac{y^2}{8}+\frac{3y^3}{32}+\cdots
    ~,~~~~~~
    y=\l^2 a~.
    \label{seriesEBI}
\end{align}
So far, this is all just described using the machinery of the auxiliary field construction. Interestingly, the interaction function (\ref{seriesEBI}) can also be found by solving the $\TT$-like flow equation
\begin{align}
    \frac{\partial \mathcal{E}(a;\lambda)}{\partial \lambda} = \frac{1}{8}\cO_{\TT} 
    = 
    \frac{1}{2}a\left(1-a(\mathcal{E}_{a})^2\right)^2
    -\frac{1}{2}\left(\mathcal{E}-2a\mathcal{E}_{a}\right)^2
    ~,
\end{align}
with the ansatz
\begin{align}
    \mathcal{E}(a) = \frac{1}{\lambda}\sE(y)
    ~.
\end{align}
The solution for the function $\sE(y)$ is
\begin{align}
    \sE(y) = \frac{y}{2} -\frac{y^2}{8} +\frac{3 y^3}{32} -\frac{13 y^4}{128} + \frac{17 y^5}{128} + \mathcal{O} \left( y^6 \right)
    ~.
\end{align}
This solution exactly reproduces the solution given in equation (\ref{seriesEBI}). There is also a method to obtain $\cH(b)$ in the $\mu$ representation, however it is merely a limiting case of the solution to $\gamma$BI and as such we postpone presenting this method until the next subsection.

One might have expected that the $\TT$ flow would yield the Born-Infeld theory in the auxiliary field formulation. However, one can also define flows that are driven by other operators. One example is rescaling the variable $b$ by a dimensionless parameter $r$ in the interaction term
\begin{align}
    \mathcal{H}(b;\lambda, r) = \frac{1}{\lambda}\frac{2rb}{rb-1}
    ~.
\end{align}
One can repeat the same steps that led to equation (\ref{BIflow}) and find that the Lagrangian satisfies the flow equation
\begin{align}
    \frac{\partial \mathcal{L}}{\partial r} = \frac{\Theta(8-\lambda\Theta)}{32r}
    ~.
\end{align}
The example given above of rescaling the variable $b$ is a simpler version of the $\sqrt{\TT}$-like deformation. A simple example of $\sqrt{\TT}$ in this formalism can be seen using another well studied theory, this being ModMax. In the standard presentation, without any auxiliary fields, the ModMax theory \cite{Bandos:2020jsw} is described by the Lagrangian
\begin{align}
    \cL_{MM}(\varphi,\Bar{\varphi}) = -\frac{\cosh(\gamma)}{2}(\varphi+\Bar{\varphi})+\sinh(\gamma)\sqrt{\varphi\Bar{\varphi}}
    ~.
\end{align}
The ModMax theory is the unique duality-invariant and conformally-invariant extension of the Maxwell theory (see \cite{Bandos:2020jsw} and appendix A of \cite{Kuzenko:2021cvx}). 
Due to conformal invariance, the $\mu$ representation is not the correct setting to study ModMax. The auxiliary field Lagrangian for ModMax in the $\nu$ representation is
\cite{Kuzenko:2021cvx}
\begin{align}
    \mathcal{L}_{MM}(\varphi,\nu) = \frac{1}{2}(\varphi+\Bar{\varphi})+\nu+\Bar{\nu}-2(V\cdot F+\Bar{V}\cdot\Bar{F})+2\tanh\left(\frac{\gamma}{2}\right)\sqrt{\nu\Bar{\nu}}
    ~.
    \label{modmaxlagrangian}
\end{align}
As it is already well known that ModMax arises as a $\sqrt{\TT}$ deformation of Maxwell theory \cite{Babaei-Aghbolagh:2022uij,Ferko:2022cix}, 
it is natural to check whether this remains true in the auxiliary field formulation. Indeed, the following flow equation is satisfied
\begin{align}
    \frac{\partial \mathcal{L}_{MM}}{\partial \gamma} = 
     \frac{1}{2} \sqrt{ T^{\mu \nu} T_{\mu \nu} - \frac{1}{4} \left( \tensor{T}{^\mu_\mu} \right)^2 }
    ~.
     \label{mmttbarflow}
\end{align}
Once again, initially, this was not found by using the auxiliary field machinery, but this can also be derived by solving the following flow equation
\begin{align}
    \frac{\partial \mathcal{E}(a;\gamma)}{\partial \gamma} = \sqrt{a}\left(1-a(\mathcal{E}_{a})^2\right)
    ~\label{mmgeneralttbarflow},
\end{align}
with the ansatz
\begin{align}
    \mathcal{E}(a;\gamma) = f(\gamma)\sqrt{a}
    ~,\quad f(0) = 0
    ~.
\end{align}
The factorisation of $\cE(a;\gamma)$ in the ansatz above is due to the conformal invariance of the model, which restricts $\cE(a;\gamma)$ to be homogeneous of degree $\frac{1}{2}$ in the variable $a$, meaning: $a\mathcal{E}_a=\frac{1}{2}\mathcal{E}$. This can be seen by setting the trace of the energy-momentum tensor to zero in equation \eqref{trace}.
Note that equation (\ref{mmgeneralttbarflow}) is simply equation (\ref{mmttbarflow}) without knowing the exact form of $\mathcal{E}(a;\gamma)$. Solving the above equation for $f(\gamma)$, one finds
\begin{align}
    f(\gamma) = 2\tanh\left(\frac{\gamma}{2}\right)
    ~,
\end{align}
as expected from equation (\ref{modmaxlagrangian}). 

We note in passing that the auxiliary field representation of the ModMax theory presented here, as well as its definition via a stress tensor flow equation, are well-defined for either sign of the deformation parameter $\gamma$. However, it was already pointed out in \cite{Bandos:2020jsw} that the ModMax theory allows for superluminal propagation when $\gamma < 0$ and only has physically sensible, causal plane wave solutions when $\gamma > 0$. This restriction on the sign of $\gamma$ is an additional physical input which is not visible at the level of the analysis that we are pursuing here. The asymmetrical behavior of the theory between the two sign choices for $\gamma$ is reminiscent of the $\TT$ deformation of a $2d$ CFT, which has a sensible spectrum for a range of positive deformation parameters $\lambda$, but for $\lambda < 0$ has infinitely many complex energy levels.\footnote{In some situations, these complex energies can be removed by performing sequential $\TT$ flows with a combination of negative and sufficiently large positive deformation parameters \cite{Ferko:2022dpg}.}

\subsection{$\gamma$-Born-Infeld}\label{subsec:gamma_BI}
We now turn our attention to the amalgamation of the previous two examples (ModMax and Born-Infeld). It is known in the literature that $\gamma$BI simultaneously obeys a $\TT$ and a $\sqrt{\TT}$ flow equation \cite{Babaei-Aghbolagh:2022uij,Ferko:2022iru,Ferko:2023ruw}. Furthermore, these two flows are commuting, which means that $\gamma$BI is connected with ModMax and Born-Infeld as per figure \ref{fig:electroflowdiagram}.
\begin{figure}[h]
    \centering
    \includegraphics[scale = 0.65]{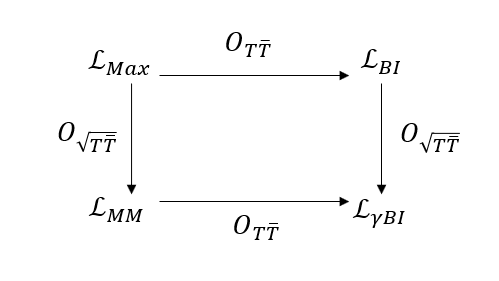}
    \caption{Flow diagram relating theories of electrodynamics}
    \label{fig:electroflowdiagram}
\end{figure}
Although the limiting theories of $\gamma$BI have been well understood in the auxiliary formalism for some time, $\gamma$BI itself had not yet been successfully elevated to an auxiliary field theory. Here, we will remedy this fact and fully explore the different facets of $\gamma$BI in the auxiliary formulations. The difficulty with this model is that it does not appear to be possible to get an explicit expression for $\mathcal{E}(a;\lambda,\gamma)$ or $\mathcal{H}(b;\lambda,\gamma)$ by trying to use the equations coming from the auxiliary field approaches. This means that one must resort to other methods of attacking the problem. 
One such method could be to take inspiration from the previous section and attempt to derive the interaction function via the $\TT$ and $\sqrt{\TT}$ flow equations. As we will now describe, this turns out to be a successful approach. Firstly, we look at the $\sqrt{\TT}$ flow equation
\begin{align}
    \frac{\pa \cH}{\pa\gamma} = \sqrt{b}\cH_{b}(1-b)
    ~.\label{rootflow}
\end{align}
Clearly, the only solution satisfying the initial condition $\mathcal{H}(b;0)=0$ is
\begin{align}
    \mathcal{H}(b;\gamma) = 0
    ~,\quad \forall \gamma,b~.
\end{align}
This is a signature of the fact that the $\mu$ representation is singular for conformal theories; if we begin with $\mathcal{H}(0;\gamma)=0$ then this is a fixed point of the $\sqrt{\TT}$ flow and one would never reach the ModMax theory. If instead, we assumed that $\mathcal{H}(0;\gamma)\neq 0$, then using the method of characteristics, the general solution to equation (\ref{rootflow}) is given by
\begin{align}
    \mathcal{H}(b;\gamma) = g(B)
    ~,\quad
B=\frac{(1+b)\cosh(\gamma)+2\sqrt{b}\sinh(\gamma)}{1-b} ~.
    \label{funcfamily}
\end{align}
Whilst this does restrict the functional form of $\mathcal{H}(b;\gamma)$, any function $g$ of the composite variable $B$ given above is a valid solution. This is all the mileage one can get from the $\sqrt{\TT}$ flow. We now turn our attention to the $\TT$ flow, under which $\mathcal{H}$ obeys the differential equation
\begin{align}
    \frac{\mathcal{H}(b;\lambda)}{\partial\lambda} = -\frac{1}{2}\mathcal{H}^2+\frac{1}{2}b\mathcal{H}^2_{b}(1-b)^2
    ~.
\end{align}
As we have seen previously, given that $\lambda$ is the only scale in the problem, one can factorise the interaction function in the following way
\begin{align}
    \mathcal{H}(b;\lambda) = \frac{1}{\lambda}\sH(b).
\end{align}
Substituting this into the previous differential equation, we find that the general solution is
\begin{align}\label{modmax_bi_mu_frame_answer}
    \sH(b) = \frac{b-1+(1+b)\cosh(c)\pm2\sqrt{b}\sinh(c)}{b-1}
    ~,
\end{align}
where $c$ is a constant of integration. Remarkably, by choosing the positive branch we obtain a candidate in the family of functions predicted by equation (\ref{funcfamily}). Specifically, 
identifying $c=\gamma$, this solution corresponds to the choice
\begin{align}\label{modmax_bi_mu_frame_answer_explicit}
    g(B)= \frac{1}{\lambda} \left( 1-B \right) \; \longrightarrow \; \mathcal{H}_{\gamma \text{BI}} (b;\gamma) = \frac{1}{\lambda} \frac{b-1+(1+b)\cosh( \gamma ) + 2\sqrt{b}\sinh(\gamma)}{b-1} \, .
\end{align}
The interaction function $\mathcal{H}_{\gamma \text{BI}}$ of equation (\ref{modmax_bi_mu_frame_answer_explicit}) is our final result for the novel $\mu$-frame representation of the ModMax-Born-Infeld theory. Note that the $\gamma = 0$ case correctly reproduces the Born-Infeld interaction function. With this, one can check that this solution indeed reproduces $\gamma$BI after integrating out the auxiliary field in the $\mu$ frame.

Of course, one is free to convert this to the $\nu$ representation, which can be done at least perturbatively. The defining relations in the $\nu$ representation are
\begin{gather}
    \mathcal{E} (  a ) = \frac{1}{\lambda} \frac{ ( b - 1 )^2 + ( b ( 4 + b ) - 1 ) \cosh ( \gamma ) + 4 b^{3/2} \sinh ( \gamma ) }{( b - 1 )^2} ~,
    \end{gather}
    where $b=b(a)$ is a solution of the following equation (which is no longer quartic):
    \begin{gather}
    \lambda^2 a = \frac{ \left( 2 \sqrt{b} \cosh ( \gamma ) + ( 1 + b ) \sinh ( \gamma ) \right)^2}{ ( b - 1 )^4 } ~.
\end{gather}
The perturbative solution for $\cE(a;\lambda,\gamma)$ is then given by
 \begin{align}
     \mathcal{E} (a;\lambda,\gamma) = 2 \sqrt{a} \tanh \left( \frac{\gamma}{2} \right) + \frac{\lambda}{2} a \sech^4 \left( \frac{\gamma}{2} \right) - \frac{1}{2} \lambda^2 a^{3/2} \sech^6 \left( \frac{\gamma}{2} \right) \tanh \left( \frac{\gamma}{2} \right) + \cdots 
     ~.
 \end{align}
Note that the initial condition for expansion is no longer $b = 0$ and $a = 0$. Maintaining that we want to obtain expressions for weak field strengths (small $a$), then one actually has to expand around $b = \tanh^2\left(\frac{\gamma}{2}\right)$ as it is at this point that $a = 0$.  For completeness, one can also obtain the above solution by solving the respective flow equations perturbatively with the interaction function $\mathcal{E}(a)$ instead of $\cH(b)$. Following the claim in Section \ref{sec:generalargument}, it is worth pointing out that one can again express the flow equations purely in terms of $\Theta$:
\bsubeq
\begin{gather}
    \frac{\partial \mathcal{H}(\Theta;\lambda,\gamma)}{\partial \lambda} = -\frac{1}{4\lambda}\Theta
    ~,
    \end{gather}
    and
    \begin{gather}
    \frac{\partial \mathcal{H}(b;\lambda,\gamma)}{\partial \gamma} 
    = 
    \frac{
    (4 \cosh (\gamma)-\lambda \Theta+4)^2 
    \left(\mathcal{K}(\Theta,\gamma)
    +2\coth(\gamma)\sqrt{\mathcal{K}(\Theta,\gamma)}
    +1\right)
    }
   {
   8 \lambda \left(-\sqrt{\lambda \Theta  (\lambda \Theta-8)}+\coth (\gamma) (\lambda \Theta-4)-4\right)
   }
    ~,
\end{gather}
   \begin{gather}
    \mathcal{K}(\Theta,\gamma)
 :=
 \frac{-8 \sqrt{\lambda \Theta \sinh ^2(\gamma) (\lambda \Theta-8)}+16 \sinh ^2(\gamma)+\lambda \Theta (\lambda \Theta-8)}{
 (4 \cosh (\gamma)-\lambda \Theta+4)^2
 }
    ~.
\end{gather}
\esubeq
Whilst the expression above is unwieldy, it provides yet another example that these flows can always be written in terms of stress tensor structures, and not necessarily in a unique way. One can obtain the flow equation for the rescaling parameter $b\rightarrow rb$ by merely making this variable replacement when solving for $b(\Theta)$.

\subsection{Some other examples}

The simplest interaction (SI) model first appeared in this context in \cite{Ivanov:2013prd} and is aptly named as both the interaction functions are merely linear in the interaction variable:
\begin{align}
    \mathcal{H}(b;L^4) = \frac{1}{L^4}b,\quad \mathcal{E}(y) = -\frac{1}{L^4}y,\quad y = L^8 a
    ~.
\end{align}
We have already discussed this model when we discussed a function $\cH$ up to quadratic in $b$. Here, we seek a ModMax extension of the case which is purely linear in $b$.
As we have already discussed before, despite having such a simple interaction function, it is not possible to obtain a closed form expression for the non-auxiliary model as this involves solving a fifth order polynomial. Naturally, we have the trace-flow equation for the parameter $L^4$:
\begin{align}
    \frac{\partial \mathcal{H}(b(\Theta);L^4)}{\partial L^4} = -\frac{\Theta}{4L^4}.
\end{align}
Interestingly, as an explicit non-auxiliary Lagrangian cannot be found for the SI model, it is only possible to formulate the $\sqrt{\TT}$ deformed or modified SI model using the auxiliary formulation. Solving the $\sqrt{\TT}$ flow equation with appropriate initial condition,
\begin{align}
    \frac{\partial \mathcal{H}(b;L^4,\gamma)}{\partial \gamma} = \sqrt{b}\mathcal{H}_{b}(1-b),\quad \mathcal{H}(b;L^4,\gamma) = \frac{1}{L^4}\mathscr{H}(b;\gamma),\quad \mathscr{H}(b,0) = b \, , 
\end{align}
yields the interaction function for the modified SI model,
\begin{gather}
    \mathcal{H}(b;L^4,\gamma) = \frac{b-1+\cosh^2\left(\frac{\gamma}{2}\right)+2\sqrt{b}\cosh\left(\frac{\gamma}{2}\right)\sinh\left(\frac{\gamma}{2}\right)+b\sinh^2\left(\frac{\gamma}{2}\right)}{L^4\left(\cosh\left(\frac{\gamma}{2}\right)+\sqrt{b}\sinh\left(\frac{\gamma}{2}\right)\right)^2} \, .
\end{gather}
One can of course consider rescaling the interaction variable as done in the previous section, however, due to the simple nature of the interaction function, this is simply equivalent to scaling the trace.

The natural progression from the previous example is to include higher degree polynomial terms in the interaction function. One can consider the same quadratic interaction function as in Section \ref{simple_instructive},
\begin{align}
    \mathcal{H}(b; L^4) = \frac{1}{L^4}(\gamma_{1} b+\gamma_{2}  b^2).
\end{align}
One can of course transition to the $\nu$ frame, however, the form does not carry over as transparently as in the linear case and hence is not very instructive. Unlike in Section \ref{simple_instructive}, we now consider rescaling the variable $b\rightarrow rb$ by a dimensionless parameter $r$. This gives the following interaction function
\begin{align}
    \mathcal{H}(b;L^4, r) = \frac{1}{L^4}(\gamma_{1} rb+\gamma_{2} r^2 b^2).
\end{align}
From this, one can obtain the following flow equation for the parameter $r$
\begin{align}
    \frac{\partial \mathcal{H}(b;L^4, r)}{\partial r} =\pm\frac{\gamma_{1} \sqrt{r^2 \left(\gamma_{1}^2+\gamma_{2} L^4 \Theta\right)}}{2 \gamma_{2} L^4 r^2}+\frac{\gamma_{1}^2}{2 \gamma_{2} L^4 r}+\frac{\Theta}{2 r}.
\end{align}
We note that in all cases thus far the flows can be written in terms of stress tensor structures as per the conclusion in Section \ref{sec:generalargument}.

Finally, one can consider an interaction function that is homogeneous of degree $n$ in the interaction variable
\begin{align}
    \mathcal{H}(b;L^4) = \frac{1}{L^4}b^n,\quad \mathcal{E}(a,L^4) = \frac{1}{L^4}(1-2n)\left(\frac{L^8a}{4n^2}\right)^{\frac{n}{2n-1}}.
\end{align}
Similarly to the linear case ($n=1$) we will ignore the rescaling flow, as this again amounts to a rescaling of the trace. Instead, we will look for the modified version of this theory by solving the $\sqrt{\TT}$ flow equation for a dimensionless parameter $\gamma$
\begin{align}
    \frac{\partial \mathcal{H}(b;L^4,\gamma)}{\partial \gamma} = \sqrt{b}\mathcal{H}_b(1-b),\quad \mathcal{H}(b;L^4,0) = \frac{1}{L^4}b^n.
\end{align}
The solution to this equation gives the modified homogeneous model
\begin{align}
    \mathcal{H}(b;L^4,\gamma) = \frac{1}{L^4}\tanh\left(\frac{1}{2}\left(-\gamma-2\text{arctanh}\left(\sqrt{b}\right)\right)\right)^{2n}.
\end{align}
The four solutions for $b(\Theta)$ are as follows
\begin{gather}
    b = \bigg\{\tanh ^2\left(\frac{\gamma}{2}\pm\tanh ^{-1}\left(\sqrt{4^{-1/n} \left(L^4 \Theta\right)^{1/n}}\right)\right),\nonumber \\
~~~~~~~~~~~~    \tanh ^2\left(-\frac{\gamma}{2}+ \tanh ^{-1}\left(\sqrt{4^{-1/n} \left(L^4 \Theta\right)^{1/n}}\right)\right)\bigg\},
\end{gather}
where the positive branch of the first solution occurs with multiplicity two. With this, the flow equation for the parameter $\gamma$ is given by
\bsubeq
\begin{align}
    \frac{\partial \mathcal{H}(b;L^4,\gamma)}{\partial \gamma} = \frac{2 n \tanh ^{2n}\left(z\right)\text{csch}\left(2z\right)}{L^4},\\
    z = \gamma+\tanh ^{-1}\left(\sqrt{4^{-1/n} \left(L^4 \Theta\right)^{1/n}}\right).
\end{align}
\esubeq
If not obvious from the preceding discussion of the linear case, obtaining on-shell models for these higher-order models is also not possible. However, it should by now be clear how the claim of Section \ref{sec:generalargument} is realized explicitly in several examples.

\section{Conclusion}\label{sec:conclusion}

In this work, we have investigated the relationship between duality invariance and stress tensor deformations in theories of electrodynamics in four spacetime dimensions. We have found that these two ideas are closely linked, and that one can think of \emph{any} family of duality-invariant theories as obeying some generalized $\TT$-like flow equation. 
A related fact, as we have seen, is that any duality-invariant function $f ( S, P )$, in a given theory $\mathcal{L} ( S, P )$ which enjoys electric magnetic duality-invariance, can be expressed as a function of the energy-momentum tensor for that fixed theory, which we express as $f ( S, P ) = f ( T_{\mu \nu} )$. Furthermore, the two Lorentz scalars that one may construct from the stress tensor are in fact dependent, so that such a function $f(S, P)$ is secretly a function of only one real variable.

Although we have made some arguments using only the differential equation obeyed by a Lagrangian $\mathcal{L} (S , P )$ for a duality-invariant theory, the hidden reduction to a univariate problem is made most transparent in an auxiliary field formulation, which was the focus of our Sections \ref{sec:auxreview} - \ref{sec:examples}. There are at least two other scenarios, in spacetime dimensions other than four, in which an auxiliary field formalism of this type might be useful to make a similar reduction to a one-variable problem manifest. Here, we will briefly describe these two scenarios below. We will then conclude by commenting about the interplay between $\TT$-like flows and their geometric engineering by means of coupling to auxiliary gravitational sectors.
We believe these are all interesting directions for future research.



\subsubsection*{\ul{\it Integrable sigma models in $d = 2$}}

Many of the structures which appear in $4d$ theories of electrodynamics also appear in certain two-dimensional models. Some of the overlap between these classes of theories was discussed in \cite{Ferko:2023ruw} following analysis in \cite{Borsato:2022tmu}, which we now very briefly review.

We consider a class of $2d$ theories which resemble the principal chiral model associated with a Lie group $G$ and its Lie algebra $\mathfrak{g}$. The fundamental degree of freedom is a group-valued field $g ( x^+, x^- ) \in G$ where $x^{\pm}$ are light-cone coordinates in the two-dimensional spacetime. It is convenient to define the left-invariant Maurer-Cartan form and its pull-back,
\begin{align}\label{j_def}
    j = g^{-1} d g \, , \qquad j_\mu = g^{-1} \partial_\mu g \, .
\end{align}
The Lagrangian of the usual principal chiral model can be written in light-cone coordinates as
\begin{align}
    \mathcal{L}_{\text{PCM}} = - \frac{1}{2} \tr \left( j_+ j_- \right) \, .
\end{align}
However, we can consider a larger class of models which depend on the two independent Lorentz invariants that can be constructed from the matrix $M_{\mu \nu} = \tr \left( j_\mu j_\nu \right)$. By analogy with the two real scalars that can be constructed from a field strength $F_{\mu \nu}$ in four dimensions, we define these invariants by the relations
\begin{align}
    S = - \frac{1}{2} \tr \left( j_+ j_- \right) \, , \qquad P^2 = \frac{1}{4} \left( \tr \left( j_+ j_+ \right) \tr \left( j_- j_- \right) - \left( \tr \left( j_+ j_- \right) \right)^2 \right) \, .
\end{align}
One can then consider a generic Lagrangian $\mathcal{L} (S , P)$ which depends on these two invariants, much as we have done for theories of electrodynamics in $4d$.

The ordinary principal chiral model, $\mathcal{L} = S$, is classically integrable; its equations of motion are equivalent to the flatness of a Lax connection for any value of a spectral parameter $z$. One might ask whether other models $\mathcal{L} ( S, P )$ share this property. If the Lagrangian satisfies
\begin{align}\label{integrable_cond}
    \mathcal{L}_S^2 - \frac{2 S}{P} \mathcal{L}_S \mathcal{L}_P - \mathcal{L}_P^2 = 1 \, ,
\end{align}
then the model is also classically integrable, and its equations of motion are equivalent to flatness of a Lax connection which can be written in light-cone coordinates $x^{\pm}$ as
\begin{align}\label{lax}
    \mathfrak{L}_{\pm} = \frac{j_{\pm} \pm z \mathfrak{J}_{\pm}}{1 - z^2} \, , 
\end{align}
for any $z \in \mathbb{C}$, where $\mathfrak{J}_\mu$ is the Noether current for invariance of the theory under right-multiplication of (\ref{j_def}) by an element $g \in G$.

We therefore see that the condition (\ref{integrable_cond}) for the model to be classically integrable, with Lax connection given by (\ref{lax}), is \emph{identical} to the differential equation obeyed by a Lagrangian for a $4d$ theory of duality-invariant electrodynamics.

It would be very interesting to construct auxiliary field formulations, much like the Ivanov-Zupnik $\mu$ and $\nu$ representations, for this class of $2d$ integrable sigma models. Because the structures are so similar, one might expect that many of the results in the present work would have analogues in the $2d$ setting. For instance, one can check that any deformation of a PCM-like model $\mathcal{L} ( S, P )$ obeying (\ref{integrable_cond}) by a function of the stress tensor preserves classical integrability; the case of root-$\TT$ deformations was investigated in \cite{Borsato:2022tmu}.

It is also known \cite{Costello:2019tri} that $2d$ integrable sigma models can be constructed from the $4d$ Chern-Simons theory which was initially studied in \cite{nekrasov_thesis,Costello:2013zra}; see \cite{Lacroix:2021iit} and references therein for a review. The relationship between the $\TT$ deformation and $4d$ Chern-Simons has been investigated in \cite{Py:2022hoa}. It would be interesting to see whether there is a $4d$ Chern-Simons construction of the $2d$ PCM-like models satisfying (\ref{integrable_cond}) and whether an auxiliary field representation exists in this setting. Perhaps one could use this to establish other properties of these sigma models, such as their one-loop structure and behavior under renormalization group flows \cite{Levine:2021mxi,Levine:2022hpv}.

\subsubsection*{\ul{\it Tensor theories in $d = 6$}}

Another setting in which many of the structures of $4d$ non-linear electrodynamics have natural analogues is among the class of six-dimensional theories of a two-form potential $A_2$ with a three-form field strength $F_3 = d A_2$. For instance, the two-parameter family of $4d$ ModMax-Born-Infeld theories -- which are relevant for the present context, in part, because they are duality invariant and satisfy commuting $\TT$-like and root-$\TT$-like flow equations -- lifts to a related family of $6d$ tensor theories \cite{Bandos:2020hgy}.

It is of particular interest to focus on \emph{chiral} theories of $2$-form electrodynamics, such as the one describing the M5-brane theory. Although one can formulate such theories using a Lorentz-invariant Lagrangian \cite{Pasti:1995tn,Pasti:1996vs,Pasti:1997gx}, it is convenient to use the Hamiltonian formalism. In the notation of \cite{Bandos:2023yat}, let us define the magnetic two-form $B^{ij} = \frac{1}{2} \epsilon^{ijklm} \partial_k A_{lm}$, where Latin indices like $i, j$ run over spatial directions $1, \cdots, 5$. Then a generic Hamiltonian density $\mathcal{H}$ for such a theory can be written as $\mathcal{H} ( s, p^2 )$,\footnote{Do not confuse the Hamiltonian $\cH$ here with the function $\cH(b)$ in the $4d$ auxiliary field $\mu$-frame.} where
\begin{align}
    s = \frac{1}{4} B^{ij} B_{ij} \, , \qquad p^2 = p^i p_i \, , \qquad p_i = \frac{1}{8} \epsilon_{ijklm} B^{jk} B^{lm} \, .
\end{align}
Lorentz invariance is not manifest in this formulation, but will be respected if the Hamiltonian density obeys
\begin{align}\label{tensor_eqn}
    \mathcal{H}_s^2 + 4 s \mathcal{H}_s \mathcal{H}_{p^2} + 4 p^2 \mathcal{H}_{p^2} = 1 \, .
\end{align}
Equation (\ref{tensor_eqn}) has the structure of the equation for duality invariance in $4d$ electrodynamics. Just as one can introduce an auxiliary field to make $4d$ duality invariance manifest, it is well-known that one can introduce an auxiliary field to make $6d$ Lorentz invariance manifest using the PST formalism. In this case, much like the $4d$ setting, theories which obey (\ref{tensor_eqn}) can also be described by interaction functions of a single variable rather than two variables $(s, p^2)$, and the energy-momentum tensor for such theories has been studied \cite{Bandos:1997gm}. By analogy with the $4d$ setting, one might expect that families of Lorentz-invariant $6d$ theories of chiral $2$-form electrodynamics may also be related by generalized stress tensor flows.\footnote{One of the main results of \cite{Bandos:2023yat} is a second condition on $\mathcal{H} ( s, p^2 )$ to guarantee that the theory exhibit zero trirefringence. It is natural to expect that, unlike the condition (\ref{tensor_eqn}) which should be preserved by all stress tensor deformations, the condition that a stress tensor flow preserve the zero-trirefringence condition should uniquely fix a single $\TT$-like flow, as in the $4d$ zero-birefringence context \cite{Ferko:2023ruw}.}

Another way to see that theories of a chiral $2$-form in six dimensions should be described by a Lagrangian that depends on one real variable, much like duality-invariant theories of $4d$ electrodynamics, is via the approach of \cite{Avetisyan:2022zza}. There the authors show that there exists only one functionally independent scalar that can be constructed from a self-dual three-form in six dimensions, which in their notation is called $I_4^{(6d)}$. A general interacting theory of a chiral $2$-form is therefore specified by an interaction term in the Lagrangian which depends on $I_4^{(6d)}$, much like the interaction functions $\mathcal{E} ( a )$ or $\mathcal{H} ( b )$ in the Ivanov-Zupnik formalism. In this language, one can describe the ModMax-like chiral tensor theory using an interaction function with the same schematic form as (\ref{modmaxlagrangian}) in the $4d$ electrodynamics setting, namely an interaction proportional to $\tanh \left( \frac{\gamma}{2} \right) \sqrt{ I_4^{(6d)} }$. It seems likely that one can also develop a $\mu$-frame version of this formalism and write an auxiliary field representation of the two-parameter family of ModMax-Born-Infeld like chiral tensors using an interaction function of the form (\ref{modmax_bi_mu_frame_answer}). We hope to revisit this direction in future work.

\subsubsection*{\ul{\it Geometric realisation of $\TT$-like deformations in $d>2$}}

Above, we have commented about two avenues where an auxiliary field sector is implemented to make manifest dynamical properties of interesting models.
In two space-time dimensions, the $\TT$ deformation has been proven in various works to possess different types of geometric interpretations.
Early in 2018 Cardy noticed that the $T\bar T$ deformation can be interpreted as coupling the original two-dimensional quantum field theory to a random geometry \cite{Cardy:2018sdv}. 
A related connection between $T\bar T$ deformations and $2d$ gravity was pushed forward in \cite{Dubovsky:2017cnj,Dubovsky:2018bmo} where it was proposed that $T\bar T$ deforming a $2d$ QFT is equivalent to coupling the theory to a Jackiw-Teitelbolm (JT) like gravity. In \cite{Tolley:2019nmm}, the deformation was interpreted as arising from a coupling to $2d$ massive gravity. See also  \cite{Caputa:2020lpa} and references within for further developments of these ideas. Geometric approaches were then used to implement algorithms to integrate several flow equations, including the Lagrangian flows \cite{Conti:2018tca,Conti:2018jho,Coleman:2019dvf,Conti:2022egv}. Moreover, in a series of papers, it was proven that $T\bar T$ flows can be derived by gauge fixing and TsT transformations of string actions \cite{Baggio:2018gct,Baggio:2018rpv,Frolov:2019nrr,Sfondrini:2019smd}. 
For the so-called ``good-sign'' (positive sign) of a $2d$ $\TT$ deformation, one can investigate the density of states at arbitrarily high energies obtaining an asymptotic Hagedorn behavior \cite{Dubovsky:2017cnj}. This indicates that the $\TT$ deformed theory is not a local QFT and is expected to describe the so-called little string theories that are dual to gravity theories on linear dilaton backgrounds \cite{Giveon:2017nie}. Interestingly, all these works indicate how the use of an auxiliary gravitational sector leads to $\TT$-like deformed quantum field theories.

To the best of our knowledge, geometric engineering of $\TT$-like flows in dimensions other than two has not been systematically pursued yet. An inspiring first analysis has been done in \cite{Conti:2022egv}, where the ModMax-Born-Infeld theory of electrodynamics was constructed in four dimensions as a geometric $\TT$-like flow. Another interesting recent reference \cite{Floss:2023nod} showed how the same models of non-linear electrodynamics that we have discussed in our paper result from integrating out massive gravitons. The known relationship between $2d$ $\TT$ and massive gravity \cite{Tolley:2019nmm}, and these recent papers, might indicate a link between general $\TT$-like flows and coupling to (massive) gravity in four dimensions which waits to be unravelled -- at least for general theories of non-linear electrodynamics. It is then an interesting avenue to explore more geometric formulations of general $\TT$-like flows in $d>2$.

\subsection*{Acknowledgements}
\noindent
C.F. is supported by U.S. Department of Energy grant DE-SC0009999 and by funds from the University of California. 
The work of S.M.K. is supported in part by the Australian Research Council, project No.~DP230101629.
L.S. is supported by a postgraduate scholarship at the University of Queensland.
The work of G.T.-M. is supported by the Australian Research Council (ARC) Future Fellowship FT180100353, and by the Capacity Building Package of the University of Queensland.
C.F. and G.T.-M. thank the participants of the Filicudi workshop on Integrability in lower-supersymmetry systems for stimulating discussions.

\appendix

\section{Details of computations for duality-invariant theories}
\label{Appendix-A}

In order to streamline the discussion in the main body of the paper, here we collect the details of a few calculations whose results were quoted without proof in Section \ref{sec:duality_and_flows}. These results all concern theories of duality-invariant electrodynamics whose Lagrangians $\mathcal{L} ( S, P )$ are written in the conventional form, in terms of the two Lorentz invariant $S$ and $P$ of equation (\ref{SandP}), rather than in one of the representations using auxiliary fields that are discussed in later sections of the paper. All of the observations in this Appendix follow from elementary considerations of the partial differential equation (\ref{EM_duality_pde}) obeyed by the Lagrangian of such self-dual theories.

\subsection{Proof that $\TT$-like flows preserve $U(1)$-duality invariance}\label{app:inductive_proof}

The goal of this Appendix is to review and expand upon the proof that any deformation of a duality-invariant Lagrangian by a function of the energy-momentum tensor preserves duality invariance. The idea of this proof was sketched in \cite{Ferko:2023ruw} which explained the result at leading order in the deformation parameter: if a duality-invariant seed theory $\mathcal{L}_0$ is deformed as
\begin{align}
    \mathcal{L}_0 \longrightarrow \mathcal{L}^1 = \mathcal{L}_0 + \l f \left( T_{\mu \nu}^{(0)} \right) \, , 
\end{align}
where $f \left( T_{\mu \nu}^{(0)} \right)$ is a function of the stress tensor for the seed theory, then the deformed theory $\mathcal{L}^1$ remains duality-invariant to order $\l$.

In fact, a similar conclusion holds to all orders in $\l$. Roughly speaking, this is because the leading-order argument can be iterated, since now the deformed theory $\mathcal{L}^1$ can be viewed as a new seed theory, and similar reasoning shows that a further deformation by a function of the first-order deformed stress tensor $T_{\mu \nu}^{(1)}$ will preserve duality-invariance at $\mathcal{O} ( \l^2 )$. This intuition that the duality invariance of the deformation can be ``bootstrapped up'' order-by-order will be made quantitative in equation (\ref{inductive_claim}) shortly. Continuing in this way, one concludes that the solution to the differential equation

\begin{align}
    \frac{\partial \mathcal{L}^{(\l)}}{\partial \l} = f \left( T_{\mu \nu}^{(\l)} \right) \, , 
\end{align}
yields a one-parameter family of functions $\mathcal{L}^{(\l)}$ which satisfy the duality-invariance condition to all orders in $\l$.

However, a more careful proof of this statement requires an inductive argument that the deformed Lagrangian $\mathcal{L}^{(\l)}$ satisfies the duality-invariance condition to all orders in $\l$. We now state this claim more precisely and spell out the reasoning in some detail. 
Note the following theorem is a particular case of Theorem \ref{theorem_preserve_invariance} and an alternative, and simpler, proof was given there. The reader should intend this subsection to be self-contained and extend on the original analysis of \cite{Ferko:2023ruw}.

\begin{theorem}\label{appendix_theorem}

Let $\mathcal{L}_0 ( S, P )$ be a Lagrangian which satisfies the duality-invariance condition given in equation (\ref{EM_duality_pde}). Suppose that there is a one-parameter family of Lagrangians $\mathcal{L}^{(\l)} ( S, P )$ which obey the flow equation
\begin{align}\label{appendix_proof_flow}
    \frac{\partial \mathcal{L}^{(\l)}}{\partial \l} = f \left( T_{\mu \nu}^{(\l)} \right) \, , 
\end{align}
where $f \left( T_{\mu \nu}^{(\l)} \right)$ is a Lorentz scalar constructed from the stress tensor of $\mathcal{L}^{(\l)}$, and with the initial condition $\mathcal{L}^{(\l)} = \mathcal{L}_0$ when $\l = 0$. Then the entire family of Lagrangians $\mathcal{L}^{(\l)}$ satisfy the same duality-invariance condition at any value of $\l$.

\end{theorem}

We note that a stronger version of this theorem is also true: one may replace the deforming operator on the right side of equation (\ref{appendix_proof_flow}) with a more general function $f ( T_{\mu \nu}^{(\lambda)} ; \lambda )$. However, for simplicity we will restrict ourselves to the case where the function depends on $\lambda$ only implicitly through the stress tensor.

\begin{proof}

We assume that the Lagrangian has a convergent Taylor series expansion in $\l$,
\begin{align}
    \mathcal{L}^{(\l)} = \mathcal{L}_0 + \l \mathcal{L}_1 + \l^2 \mathcal{L}_2 + \cdots \, .
\end{align}
We use the symbols $\mathcal{L}_i$ with a lower index for the Taylor coefficients in the Lagrangian, in contrast to the variables $\mathcal{L}^{k}$ with an upper index, which we define as the approximation to $\mathcal{L}^{(\l)}$ which is accurate up to $\mathcal{O} ( \l^k )$,
\begin{align}
    \mathcal{L}^{k} = \sum_{i=0}^{k} \l^i \mathcal{L}_i \, .
\end{align}
Likewise, we let $T_{\mu \nu}^{k}$ be the energy-momentum tensor constructed from $\mathcal{L}^k$. By virtue of the differential equation (\ref{appendix_proof_flow}), the approximate Lagrangians $\mathcal{L}^k$ satisfy
\begin{align}\label{taylor_to_stress}
    \mathcal{L}^{k+1} = \mathcal{L}^{k} + \frac{\l^{k+1}}{k+1} \Big[ f \left( T_{\mu \nu}^{k} \right) \Big]_{\l^k} \, ,
\end{align}
where the notation $\left[ f \left( T_{\mu \nu}^{k} \right) \right]_{\l^k}$ means to extract the Taylor coefficient proportional to $\l^k$ in the series expansion of $f \left( T_{\mu \nu}^{k} \right) $. Explicitly,
\begin{align}
    \Big[ g ( \l ) \Big]_{\l^k} = \frac{1}{k!} \frac{d^k g}{d \l^k} \Big\vert_{\l = 0} \, , 
\end{align}
for any function $g(\l)$.

It is convenient to parameterize a general Lorentz scalar function $f^k = f \left( T_{\mu \nu}^{k} \right)$ in terms of the two variables
\begin{align}
    \Theta^k = \tensor{\left( T^k \right)}{^\mu_\mu} \, , \qquad \widehat{T}^{k,2} = \left( \widehat{T}^k \right)^{\mu \nu} \left( \widehat{T}^k \right)_{\mu \nu} \, , 
\end{align}
where $\widehat{T}^k_{\mu \nu}$ is the traceless part of $T^k_{\mu \nu}$. In this parameterization, we write
\begin{align}\label{appendix_invariants_defn}
    f ( T_{\mu \nu}^{k} ) &= f \left( \Theta^k , \widehat{T}^{k,2} \right) \, , \nonumber \\
    \Theta^k &= 4 \left( \mathcal{L}^{k}  - P \mathcal{L}^{k}_P - S \mathcal{L}^{k}_S \right) \, , \nonumber \\
    \widehat{T}^{k,2} &= 4\left( S^2 + P^2 \right) \left( \mathcal{L}^{k}_S\right)^2 \, .
\end{align}
We will also collect some formulas involving derivatives of $f^k = f ( T_{\mu \nu}^k )$:
\begin{align}\label{f_derivs}
    \frac{\partial f^k}{\partial S} &= \frac{\partial f^k}{\partial \Theta^k} \frac{\partial \Theta}{\partial S} + \frac{\partial f^k}{\partial \widehat{T}^{k,2}} \frac{\partial \widehat{T}^{k,2}}{\partial S} \nonumber \\
    &= - 4 \frac{\partial f^k}{\partial \Theta^k} \left( P \mathcal{L}_{SP}^k + S \mathcal{L}_{SS}^k \right) + 4 \frac{\partial f^k}{\partial \widehat{T}^{k,2}} \left( 2 S \left( \mathcal{L}_S^k \right)^2 + 2 \left( S^2 + P^2 \right) \mathcal{L}_S^k \mathcal{L}_{SS}^k \right) \, , \nonumber \\
    \frac{\partial f^k}{\partial P} &= \frac{\partial f^k}{\partial \Theta^k} \frac{\partial \Theta}{\partial P} + \frac{\partial f}{\partial \widehat{T}^{k,2}} \frac{\partial \widehat{T}^2}{\partial P} \nonumber \\
    &= - 4 \frac{\partial f^k}{\partial \Theta^k} \left( P \mathcal{L}_{PP}^k - S \mathcal{L}_{SP}^k \right) + 4 \frac{\partial f^k}{\partial \widehat{T}^{k,2}} \left( 2 P \left( \mathcal{L}_S^k \right)^2 + 2 \left( S^2 + P^2 \right) \mathcal{L}_S^k \mathcal{L}_{SP}^k \right) \, .
\end{align}
For any function $h ( S, P )$ and any non-negative integer $k$, we also define the functions\footnote{The functions $F$ and $G$ are not to be confused with the field strength $F_{\mu \nu}$ and the quantity $G_{\mu \nu}$ introduced in equation (\ref{duality-rotation-0}), both of which carry Lorentz indices.}
\begin{align}
    F ( h ) &= \left( h_S \right)^2 - \frac{2 S}{P} h_S h_P - \left( h_P \right)^2 - 1 \, , \nonumber \\
    F^k &= F \left( \mathcal{L}^k \right) \, , \nonumber \\
    G^k ( h ) &= 2 \mathcal{L}_S^k h_S - \frac{2 S}{P} \left( \mathcal{L}_S^k h_P + \mathcal{L}_P^k h_S \right) - 2 \mathcal{L}_P^k h_P \, .
\end{align}
The function $F^k$ measures the failure of the Lagrangian $\mathcal{L}^k$ to satisfy the duality invariance condition, whereas the function $G^k ( h )$ measures the failure of the function $h ( S, P )$ to be invariant under the duality transformation associated with the Lagrangian $\mathcal{L}^k$. It is easy to see that every $F^n$ obeys a recursion relation,
\begin{align}\label{Fn_recursive}
    F^n &= \left( \mathcal{L}^{n-1}_S + \l^n \mathcal{L}_{n, S} \right)^2 - \frac{2 S}{P} \left( \mathcal{L}^{n-1}_S + \l^n \mathcal{L}_{n, S}  \right) \left( \mathcal{L}^{n-1}_P + \l^n \mathcal{L}_{n, P}  \right) - \left( \mathcal{L}^{n-1}_P + \l^n \mathcal{L}_{n, P} \right)^2 - 1 \, \nonumber \\
    &= \Big[ \left( \mathcal{L}_S^{n-1} \right)^2 - \frac{2 S}{P} \mathcal{L}_S^{n-1} \mathcal{L}_P^{n-1} - \left( \mathcal{L}_P^{n-1} \right)^2 \Big] + \l^{2n} \Big[ \mathcal{L}_{n, S}^2 - \frac{2 S}{P} \mathcal{L}_{n, S} \mathcal{L}_{n, p} - \mathcal{L}_{n, P}^2 \Big] \nonumber \\
    &\quad + \l^n \Big[ 2 \mathcal{L}_S^{n-1} \mathcal{L}_{n, S} - \frac{2 S}{P}  \left( \mathcal{L}_S^{n-1} \mathcal{L}_{n, P} + \mathcal{L}_P^{n-1} \mathcal{L}_{n, S} \right) - 2 \mathcal{L}_{P}^{n-1} \mathcal{L}_{n, P} \Big] \nonumber \\
    &= F^{n-1} + F \left( \l^n \mathcal{L}_n \right) + G^{n-1} \left( \l^n \mathcal{L}_n \right) \, .
\end{align}
The key technical step in our proof is to establish the relation
\begin{align}\label{inductive_claim}
    F^k + \sum_{i = 0}^{ k - 1} G^i \left( \l^{2k - i } \mathcal{L}_{2 k - i} \right) = 0 \, .
\end{align}
We will prove this by induction on $k$. When $k = 0$, there are no terms in the sum, so the claim reduces to the statement that $F^0 = 0$, which is automatically true by the assumption that the undeformed theory is duality invariant.

\subsubsection*{\ul{\it Base Case: $k = 1$}}

Let us consider the first non-trivial case, which is $k = 1$. In this case the claim is
\begin{align}\label{base_case_improved}
    F^1 + G^0 \left( \l^2 \mathcal{L}_2 \right) = 0 \, .
\end{align}
Using the recursive relation (\ref{Fn_recursive}), we have $F^1 = F^0 + F \left( \l \mathcal{L}_1 \right) + G^0 \left( \l \mathcal{L}_1 \right)$, and again we have $F^0 = 0$ by assumption. So we would like to show that
\begin{align}
    F \left( \l \mathcal{L}_1 \right) + G^0 \left( \l \mathcal{L}_1 \right) + G^0 \left( \l^2 \mathcal{L}_2 \right) = 0 \, .
\end{align}
Note that $F \left( \l \mathcal{L}_1 \right)$ and $G^0 \left( \l^2 \mathcal{L}_2 \right)$ are both of order $\l^2$ while $G^0 \left( \l \mathcal{L}_1 \right)$ is linear in $\l$, so we will first show that $G^0 \left( \l \mathcal{L}_1 \right) = 0$. Using the expression (\ref{taylor_to_stress}) for $\mathcal{L}_1$ in terms of $f \left( T_{\mu \nu}^0 \right) = f^0$, this means that we must show
\begin{align}\label{base_case_duality}
    \mathcal{L}_S^0 f_S^0 - \frac{S}{P} \left( f_S^0 \mathcal{L}_P^0 + \mathcal{L}_S^0 f_P^0 \right) - \mathcal{L}_P^0 f_P^0 = 0 \, ,
\end{align}
where $f^0 = f \left( T_{\mu \nu}^0 \right)$. Equation (\ref{base_case_duality}) expresses the condition that the function $f \left( T_{\mu \nu}^0 \right)$ is duality invariant with respect to the duality transformation in the undeformed theory $\mathcal{L}_0$. To prove this, we use that $F^0 ( S, P ) = 0$ identically, which means that both the function and its derivatives with respect to $S$ and $P$ are equal to zero. The conditions $\partial_S F^0 = 0$ and $\partial_P F^0 = 0$ give the constraints
\begin{align}\label{first_deriv_constraints}
    \mathcal{L}_S^0 \mathcal{L}_{SS}^0 - \frac{1}{P} \mathcal{L}_S^0 \mathcal{L}_P^0 - \frac{S}{P} \left( \mathcal{L}_{SS}^0 \mathcal{L}_P^0 + \mathcal{L}_S^0 \mathcal{L}_{SP}^0 \right) - \mathcal{L}_P^0 \mathcal{L}_{SP}^0 &= 0 \, , \nonumber \\
    \mathcal{L}_S^0 \mathcal{L}_{SP}^0 + \frac{S}{P^2} \mathcal{L}_S^0 \mathcal{L}_P^0 - \frac{S}{P} \left( \mathcal{L}_{SP}^0 \mathcal{L}_P^0 + \mathcal{L}_S^0 \mathcal{L}_{PP}^0 \right) - 2 \mathcal{L}_P^0 \mathcal{L}_{PP}^0  &= 0 \, .
\end{align}
Equations (\ref{first_deriv_constraints}) give conditions which allow us to eliminate some of the second derivative terms which arise when substituting the expressions (\ref{f_derivs}) for $f^k$, with $k = 0$, into (\ref{base_case_duality}). Explicitly, we compute
\begin{align}\label{big_eqn_base_case}
    &\mathcal{L}_S^0 f_S^0 - \frac{S}{P} \left( f_S^0 \mathcal{L}_P^0 + \mathcal{L}_S^0 f_P^0 \right) - \mathcal{L}_P^0 f_P^0 \nonumber \\
    &\quad = \mathcal{L}_S^0 \left( - 4 \frac{\partial f^0}{\partial \Theta^0} \left( P \mathcal{L}_{SP}^0 + S \mathcal{L}_{SS}^0 \right) + 4 \frac{\partial f^0}{\partial \widehat{T}^{0,2}} \left( 2 S \left( \mathcal{L}_S^0 \right)^2 + 2 \left( S^2 + P^2 \right) \mathcal{L}_S^0 \mathcal{L}_{SS}^0 \right) \right) \nonumber \\
    &\quad - \frac{S}{P} \Bigg[ \left( \left( - 4 \frac{\partial f^0}{\partial \Theta^0} \left( P \mathcal{L}_{SP}^0 + S \mathcal{L}_{SS}^0 \right) + 4 \frac{\partial f^0}{\partial \widehat{T}^{0,2}} \left( 2 S \left( \mathcal{L}_S^0 \right)^2 + 2 \left( S^2 + P^2 \right) \mathcal{L}_S^0 \mathcal{L}_{SS}^0 \right) \right) \right) \mathcal{L}_P^0 \nonumber \\
    &\quad + \mathcal{L}_S^0  \left( - 4 \frac{\partial f^0}{\partial \Theta^0} \left( P \mathcal{L}_{PP}^0 - S \mathcal{L}_{SP}^0 \right) + 4 \frac{\partial f^0}{\partial \widehat{T}^{0,2}} \left( 2 P \left( \mathcal{L}_S^0 \right)^2 + 2 \left( S^2 + P^2 \right) \mathcal{L}_S^0 \mathcal{L}_{SP}^0 \right) \right) \Bigg] \nonumber \\
    &\quad - \mathcal{L}_P^0 \left( - 4 \frac{\partial f^0}{\partial \Theta^0} \left( P \mathcal{L}_{PP}^0 - S \mathcal{L}_{SP}^0 \right) + 4 \frac{\partial f^0}{\partial \widehat{T}^{0,2}} \left( 2 P \left( \mathcal{L}_S^0 \right)^2 + 2 \left( S^2 + P^2 \right) \mathcal{L}_S^0 \mathcal{L}_{SP}^0 \right) \right) \, .
\end{align}
After substituting the constraints (\ref{first_deriv_constraints}) into equation (\ref{big_eqn_base_case}), simplifying using the condition that $F^0 = 0$ due to the duality invariance of the seed theory $\mathcal{L}^0$, and performing some algebra, one finds that all dependence on the derivatives of $f^0$ drops out, and
\begin{align}
    \mathcal{L}_S^0 f_S^0 - \frac{S}{P} \left( f_S^0 \mathcal{L}_P^0 + \mathcal{L}_S^0 f_P^0 \right) - \mathcal{L}_P^0 f_P^0 = 0 
\end{align}
holds identically, regardless of the value of $\frac{\partial f^0}{\partial \Theta^0}$ and $\frac{\partial f^0}{\partial \widehat{T}^{0, 2}}$. This establishes that the terms of order $\l$ in (\ref{base_case_improved}) vanish.

Let us now consider the terms of order $\l^2$. We must now show that $F \left( \l \mathcal{L}_1 \right) + G^0 \left( \l^2 \mathcal{L}_2 \right) = 0$, or
\begin{align}
    0 &= \mathcal{L}_{1, S}^2 - \frac{2 S}{P} \mathcal{L}_{1, S} \mathcal{L}_{1, P} - \mathcal{L}_{1, P}^2 - 1 + 2 \left( \mathcal{L}^0_S \mathcal{L}_{2, S} - \frac{S}{P} \left( \mathcal{L}^0_S \mathcal{L}_{2, S} + \mathcal{L}^0_P \mathcal{L}_{2, S} \right) - 2 \mathcal{L}^0_P \mathcal{L}_{2, P} \right) \, .
\end{align}
To do this we must use the facts that
\begin{align}
    \mathcal{L}_1 = f \left( T_{\mu \nu}^0 \right) \, , \qquad \mathcal{L}_2 = \frac{1}{2} \Big[ f \left( T_{\mu \nu}^1 \right) \Big]_{\l} \, , 
\end{align}
along with our formulas (\ref{f_derivs}) for derivatives of the function $f$. In particular, it is important that the argument $T_{\mu \nu}^{1}$ of the function $f$ in $\mathcal{L}_2$ is itself determined in terms of the same function $f$:
\begin{align}
     f \left( T_{\mu \nu}^1 \right) &= f \left[ T_{\mu \nu}^{0} + \lambda T_{\mu \nu} \left( \mathcal{L}_1 \right) \right] \nonumber \\
     &= f \left[ T_{\mu \nu}^{0} + \lambda T_{\mu \nu} \left( f \left( T_{\mu \nu}^{(0)} \right) \right) \right] \, .
\end{align}
This is because the Hilbert stress tensor is a linear function of the Lagrangian, so in general for a sum $\mathcal{L} = \mathcal{L}_A + \mathcal{L}_B$, the total stress tensor is $T_{\mu \nu} ( \mathcal{L} ) = T_{\mu \nu} ( \mathcal{L}_A ) + T_{\mu \nu} ( \mathcal{L}_B )$.

Our calculation only requires us to extract the term in $f \left( T_{\mu \nu}^1 \right)$ which is proportional to $\lambda^1$, or
\begin{align}\label{L2_recursive}
    \mathcal{L}_2 = \frac{1}{2} \frac{d}{d \lambda} \left\{ f \left[ T_{\mu \nu}^{0} + \lambda T_{\mu \nu} \left( f \left( T_{\mu \nu}^{(0)} \right) \right) \right] \right\} \Big\vert_{\lambda = 0 } \, .
\end{align}
We note that all of these quantities are ultimately determined in terms of $\mathcal{L}_0$, which satisfies the exact duality-invariance condition.

We may therefore evaluate derivatives of $\mathcal{L}_2$ with respect to $S$ and $P$ using the expression (\ref{L2_recursive}) along with our previous results (\ref{f_derivs}). After doing this and simplifying using the duality invariance of $\mathcal{L}^0$, one finds that
\begin{align}
    F \left( \l \mathcal{L}_1 \right) + G^0 \left( \l^2 \mathcal{L}_2 \right) = 0 \, , 
\end{align}
which completes the proof that our claim (\ref{inductive_claim}) holds in the case $k=1$.

\subsubsection*{\ul{\it Inductive step}}

We now suppose that equation (\ref{inductive_claim}) holds for $k = 1 , \cdots , n-1$ and show that it also holds when $k = n$. Using the recursion relation (\ref{Fn_recursive}) for the $F^k$ and our induction hypothesis, we have
\begin{align}\label{inductive_step_intermediate}
    F^n = - \sum_{i = 0}^{n - 2} G^i \left( \lambda^{2 ( n - 1 )  - i } \mathcal{L}_{2 ( n - 1 ) - i} \right) + F \left( \lambda^n \mathcal{L}_n \right) + G^{n - 1} \left( \lambda^n \mathcal{L}_n \right) \, .
\end{align}
We would like to eliminate the last two terms in (\ref{inductive_step_intermediate}) and express the result entirely in terms of a sum of $G^i$ with various arguments. To do this, we must again rely on the recursive definition of the Taylor coefficients $\mathcal{L}_i$ in the Lagrangian:
\begin{align}\label{inductive_intermediate_expansion}
    \mathcal{L}_j &= \frac{1}{j} \left[ f \left( T_{\mu \nu}^{j-1} \right) \right]_{\lambda^{j-1}} \, , \nonumber \\
    T_{\mu \nu}^j &= T_{\mu \nu} \left( \mathcal{L}_0 \right) + \lambda T_{\mu \nu} \left( \mathcal{L}_1 \right) + \cdots + \lambda^j T_{\mu \nu} \left( \mathcal{L}_j \right) \, .
\end{align}
Extracting the term of order $\lambda^{j-1}$ in an expression (\ref{inductive_intermediate_expansion}),
\begin{align}
    \left[ f \left( T_{\mu \nu}^{j-1} \right) \right]_{\lambda^{j-1}} = \frac{1}{ ( j - 1 )!} \frac{d^{j-1}}{d \lambda^{j-1}} \left[ f \left( T_{\mu \nu}^{j-1} \right) \right]_{\lambda = 0} \, , 
\end{align}
then generates a series of terms involving lower $\mathcal{L}_i$ which are all defined in terms of the same expansions (\ref{inductive_intermediate_expansion}). It turns out that this recursive definition, along with the duality invariance condition for the undeformed Lagrangian $\mathcal{L}_0$, implies the relation
\begin{align}
    \left[ \sum_{i = 0}^{n - 1} G^i \left( \lambda^{2n - i } \mathcal{L}_{2 n - i} \right) \right] + F \left( \lambda^n \mathcal{L}_n \right) + G^{n - 1} \left( \lambda^n \mathcal{L}_n \right) = \sum_{i = 0}^{n - 2} G^i \left( \lambda^{2 ( n - 1 )  - i } \mathcal{L}_{2 ( n - 1 ) - i} \right) \, .
\end{align}
Combining this formula with the result (\ref{inductive_step_intermediate}) of our inductive hypothesis and the recursion relation for $F^n$, we find
\begin{align}
    F^n = - \sum_{i = 0}^{n - 1} G^i \left( \lambda^{2n - i } \mathcal{L}_{2 n - i} \right) \, , 
\end{align}
which establishes that (\ref{inductive_claim}) also holds when $k = n$. This formula therefore holds for all integers $k \geq 0$ by induction.

\subsubsection*{\ul{\it Proof of original claim}}

Now that we have established equation (\ref{inductive_claim}) by induction, let us return to the proof of the main theorem. We would like to show that the full solution $\mathcal{L}^{(\l)}$ to the flow equation is duality invariant, which in the notation developed above is expressed by the statement
\begin{align}
    F \left( \mathcal{L}^{(\l)} \right) = 0  \, .
\end{align}
Using the Taylor series expansion for $\mathcal{L}^{(\l)}$, we may write
\begin{align}\label{Fgamma_to_Fk}
    F \left( \mathcal{L}^{(\l)} \right) = \lim_{k \to \infty} F^k \, .
\end{align}
However, from (\ref{inductive_claim}) we see that
\begin{align}\label{Fk_asymptotic}
    F^k = \mathcal{O} \left( \l^{k+1} \right) \, .
\end{align}
This expresses the fact that, at each order $k$ in the Taylor series expansion $\mathcal{L}^k$ of $\mathcal{L}^{(\l)}$, the theory is duality invariant to order $\l^k$, and the failure of duality invariance begins only at order $\l^{k+1}$. Therefore, taking the $k \to \infty$ limit in (\ref{Fk_asymptotic}), we conclude that
\begin{align}
    F \left( \mathcal{L}^{(\l)} \right) = 0 \, , 
\end{align}
which proves Theorem \ref{appendix_theorem}.
\end{proof}

\subsection{Method of characteristics and $U(1)$-duality invariance}\label{app:method_char}

In this Appendix we will prove that, in a duality-invariant theory described by a Lagrangian $\mathcal{L} ( S, P )$, any function $f ( S, P )$ which is invariant under duality transformations can be expressed as a function of a single variable. This single variable can be chosen to be any non-trivial Lorentz scalar constructed from the stress tensor $T_{\mu \nu}$. Our proof will rely on the method of characteristics, which is a standard technique for solving first-order partial differential equations. See also \cite{Hou:2022csf} for another application of this method to study $\TT$-like flows.

In order for a function $f ( S, P )$ to be invariant under the duality transformations associated with the Lagrangian $\mathcal{L} ( S, P )$, this function must satisfy the differential equation
\begin{align}\label{duality_PDE}
    \left( P \mathcal{L}_S - S \mathcal{L}_P \right) f_S - \left( S \mathcal{L}_S + P \mathcal{L}_P \right) f_P = 0 \, .
\end{align}
We will first seek characteristic curves for this differential equations, which are one-parameter families of points
\begin{align}
    \left( S ( t ) , P ( t ) ,  f ( S ( t ) , P ( t ) ) \right)
\end{align}
described by a parameter $t$ which labels points along the curve. The characteristic curves satisfy the system of ordinary differential equations
\begin{align}\label{characteristic_system}
    \frac{dS}{dt} &= P \mathcal{L}_S - S \mathcal{L}_P \, , \nonumber \\
    \frac{dP}{dt} &= - S \mathcal{L}_S - P \mathcal{L}_P \, , \nonumber \\
    \frac{df}{dt} &= 0 \, ,
\end{align}
which guarantees that the tangent vector to the curve is also a tangent vector to the plane of solutions to the differential equation (\ref{duality_PDE}). Clearly solutions to the system (\ref{characteristic_system}) have the property that
\begin{align}\label{ft_constant}
    f ( t ) = u
\end{align}
for some constant $u$ which is independent of the parameter $t$. It will be helpful to look for other functions $v ( S ( t ) , P ( t ) )$ which are independent of $t$, so that
\begin{align}\label{const_char}
    0 &= \frac{d v}{dt} \nonumber \\
    &= \frac{\partial v}{\partial S} \frac{d S}{d t} + \frac{\partial v}{\partial P} \frac{d S}{d t} \nonumber \\
    &= \frac{\partial v}{\partial S} \left( P \mathcal{L}_S - S \mathcal{L}_P \right)  + \frac{\partial v}{\partial P} \left( - S \mathcal{L}_S - P \mathcal{L}_P \right) \, , 
\end{align}
where in the last step we have substituted (\ref{characteristic_system}) for $S'(t)$ and $P'(t)$. Any such function $v$ will be constant along the characteristics curves for which (\ref{characteristic_system}) holds.

We first claim that any function of the energy-momentum tensor associated with $\mathcal{L} ( S, P )$ provides us with such a function $v ( S,  P )$, assuming that the theory enjoys duality invariance. As we mentioned in Section \ref{sec:generalities}, any function of the stress tensor can be written as a function of the two Lorentz scalars $\Theta$ and $T^2$,
\begin{align}
    f ( T_{\mu \nu} ) &= f ( \Theta, T^2 ) \, , \nonumber \\
    \Theta &= 4 \left( \mathcal{L}  - P\mathcal{L}_P - S\mathcal{L}_S \right) \, , \nonumber \\
    T^2 &= 4\left( S^2 + P^2 \right) \mathcal{L}_S^2 + 4\left( \mathcal{L} - P \mathcal{L}_P - S \mathcal{L}_S \right)^2 \, .
\end{align}
Therefore, it suffices to show that the two functions $\Theta ( S, P )$ and $T^2(S,P)$ satisfy the condition (\ref{const_char}) which means that they are constant along characteristic curves. We first compute the derivatives $\frac{d \Theta}{d t}$ and $\frac{d T^2}{dt}$, assuming that $S(t)$ and $P(t)$ satisfy (\ref{characteristic_system}):
\begin{align}\label{du_dt_and_dv_dt}
    \frac{d \Theta}{d t} &= 4 \Bigg( P^2 \left( \mathcal{L}_P \mathcal{L}_{PP} - \mathcal{L}_S \mathcal{L}_{SP} \right) + P S \left( 2 \mathcal{L}_P \mathcal{L}_{SP} + \mathcal{L}_S \left( \mathcal{L}_{PP} - \mathcal{L}_{SS} \right) \right) + S^2 \left( \mathcal{L}_S \mathcal{L}_{SP} + \mathcal{L}_P \mathcal{L}_{SS} \right) \Bigg) \, , \nonumber \\
    \frac{d T^2}{dt } &= - 8 \left( S^2 + P^2 \right) \mathcal{L}_S^2 \mathcal{L}_P + 8 \left( S^2 + P^2 \right) \mathcal{L}_S \left( \mathcal{L}_{SS} \left( P \mathcal{L}_S - S \mathcal{L}_P \right) - \mathcal{L}_{SP} \left( S \mathcal{L}_S + P \mathcal{L}_{P} \right) \right) \nonumber \\
    &\quad + 8 \left( \mathcal{L} - P \mathcal{L}_P - S \mathcal{L}_S \right) \cdot \Bigg( P^2 \left( \mathcal{L}_P \mathcal{L}_{PP} - \mathcal{L}_S \mathcal{L}_{SP} \right) + P S \left( 2 \mathcal{L}_P \mathcal{L}_{SP} + \mathcal{L}_S \left( \mathcal{L}_{PP} - \mathcal{L}_{SS} \right) \right) \nonumber \\
    &\quad + S^2 \left( \mathcal{L}_S \mathcal{L}_{SP} + \mathcal{L}_P \mathcal{L}_{SS} \right) \Bigg) \, . 
\end{align}
After imposing the duality invariance condition (\ref{EM_duality_pde}), as well as the derivatives of this equation with respect to $S$ and $P$, the combinations appearing in (\ref{du_dt_and_dv_dt}) collapse to
\begin{align}
    \frac{d \Theta}{dt} = 0 \, , \qquad \frac{d T^2}{dt} = 0 \, .
\end{align}
This means that $\Theta$ and $T^2$, and therefore a general Lorentz scalar function of the stress tensor, is constant along the characteristic curves. Such functions are said to be integrals of the characteristic system.

It is a general theorem that, if two integrals $u, v$ of the characteristic system are known for a first-order linear partial differential equation for a function $f ( S, P )$ of two variables $S, P$, then the general solution to this differential equation is described implicitly by
\begin{align}
    g ( u, v ) = 0 \, , 
\end{align}
where $g$ is an arbitrary function of two independent variables. We have already seen in equation (\ref{ft_constant}) that, since $\frac{df}{dt} = 0$, the function $f(t) = u$ is one such integral of the characteristic system. In order to write down the general solution to the differential equation (\ref{duality_PDE}), we therefore only need to identify one other integral of the characteristic system -- and indeed, we are guaranteed that at most one other functionally independent quantity of this type exists. In particular, this implies that for any duality-invariant Lagrangian and any two quantities $v_1 ( T_{\mu \nu} )$, $v_2 ( T_{\mu \nu} )$ which are constructed from the stress tensor, one of the two quantities $v_1, v_2$ can be locally expressed as a function of the other. For instance, if the trace $\Theta$ is a non-trivial function of $S$ and $P$ (i.e. if $\Theta$ is not a constant), then it must be possible to express it as a function of $T^2$. We already expected that this should be true from the arguments around equation (\ref{functional_relation_theta_Tsq}) which demonstrate that there exists some functional relation of the form $h ( \Theta, T^2 ) = 0$ in any duality-invariant model.

Therefore, let us choose $v$ to be \emph{any} function of the energy-momentum tensor which is a non-trivial function of $S$ and $P$. To be concrete, we can choose $v = T^2$ since this combination $T^{\mu \nu} T_{\mu \nu}$ is non-trivial in all of the models which we will consider (unlike the trace $\Theta$, which vanishes in conformal models such as the Maxwell and ModMax theories). The general solution to the (\ref{duality_PDE}) is therefore
\begin{align}
    g \left( u , v \right) = 0 \, , 
\end{align}
for some function $g$ of two variables. By the inverse function theorem, this means that $u = f ( S, P )$ can locally be expressed as a function of $v$, which means that
\begin{align}
    f ( S, P ) = h ( v )
\end{align}
for some function $h$. For the choice $v = T^2$, we conclude that any duality-invariant function can be written as a function of the single variable $T^2 = T^{\mu \nu} T_{\mu \nu}$.

A simple example is the Maxwell Lagrangian $\mathcal{L} = S$, for which one has
\begin{align}
    \Theta = 0 \, , \qquad T^2 = 4 \left( S^2 + P^2 \right) \, .
\end{align}
In this case, our general argument shows that any duality-invariant function can be written as $f ( T^2 )$ or equivalently $f ( S^2 + P^2 )$. Note that the trace $\Theta$ is indeed functionally dependent on the other invariant $T^2$, albeit in a trivial way because it vanishes.

\subsubsection*{\ul{\it Solution to differential equation for the Lagrangian}}

The preceding argument shows that any duality-invariant function $f ( S, P )$ can be written as a function of a single variable; for instance, this variable can be taken to be $T^2$. A similar statement holds for the Lagrangian of a theory of duality-invariant electrodynamics. As we have mentioned, a Lagrangian $\mathcal{L} ( S, P )$ which described a duality-invariant theory must satisfy the partial differential equation (\ref{EM_duality_pde}).
This differential equation is similar, but not identical, to the condition (\ref{duality_PDE}) satisfied by a duality-invariant function. This reflects the fact that the Lagrangian itself need not be invariant under duality rotations in order for the equations of motion to be duality-invariant; the Maxwell theory $\mathcal{L} = S$ is a counter-example.

Another difference between (\ref{EM_duality_pde}) and (\ref{duality_PDE}) is that the differential equation for $f ( S, P )$ is linear, which allows one to solve it using the method of characteristics, whereas the equation for $\mathcal{L} ( S, P )$ is non-linear. Nonetheless, this equation can be solved and the general solution to this duality-invariance condition for $\mathcal{L}$ is also described by a function of a single variable -- see also previous discussions in \cite{Gaillard:1997rt,Hatsuda:1999ys,Ivanov:2003uj}. 

For completeness, we now briefly review the standard argument for this conclusion. It is first convenient to rewrite equation (\ref{EM_duality_pde}) in new variables. Recalling the definitions
\begin{align}
    \varphi = F^{\alpha\beta}F_{\alpha\beta},\quad \overbar{\varphi}=\overbar{F}_{\dot\alpha\dot\beta}\overbar{F}^{\dot\alpha\dot\beta} \, , 
\end{align}
introduced in Section \ref{sec:auxreview}, let us define
\begin{align}
    p = \frac{1}{4} \left( \varphi + \overbar{\varphi} \right) + \frac{1}{2} \sqrt{ \varphi \overbar{\varphi}} \, , \qquad q = \frac{1}{4} \left( \varphi + \overbar{\varphi} \right) - \frac{1}{2} \sqrt{ \varphi \overbar{\varphi}} \, .
\end{align}
In terms of these variables, the differential equation for the Lagrangian becomes
\begin{align}\label{courant_hilbert}
    \mathcal{L}_p \mathcal{L}_q = 1 \, , 
\end{align}
which is known as the Courant-Hilbert equation \cite{courant_hilbert}. The general solution to this differential equation is
\begin{align}\label{courant_hilbert_solution}
    \mathcal{L} ( p, q ) = v ( s ) + \frac{2 p}{v'(s)} \, , 
\end{align}
where $v(s)$ is an arbitrary function of one variable, and the auxiliary variable $s$ is related to the dynamical quantities $p, q$ by
\begin{align}
    q = s + \frac{p}{ \left( v' ( s ) \right)^2 } \, .
\end{align}
This makes it clear that theories of duality-invariant electrodynamics, without higher derivative interactions -- so that the Lagrangian depends on $S, P$ but not invariants involving $\partial_\rho F_{\mu \nu}$ and so forth -- are in one-to-one correspondence with functions of a single real variable $v(s)$.

One might have expected this fact from the discussion of the auxiliary field representations of Section \ref{sec:auxreview}. Indeed, \emph{any} solution to the duality-invariance condition (\ref{courant_hilbert}) also admits an auxiliary field description in the $\nu$ frame in terms of an interaction function $\cE ( a )$ where $a = \nu \overbar{\nu}$, as mentioned in \cite{Ivanov:2003uj}. Therefore, duality-invariant theories (again, without higher-derivative terms) may be viewed as being in one-to-one correspondence with univariate functions in two ways: each such theory is described by either a function $v(s)$ as in (\ref{courant_hilbert_solution}) or by a function $\cE( a )$ in the $\nu$ representation.

\subsection{Vanishing of Jacobian determinant}\label{app:determinant}

In this Appendix we will explain the brief computation which leads to the vanishing of the Jacobian determinant (\ref{jacobian_zero}) for theories of self-dual electrodynamics. We aim to compute the determinant of the matrix
\begin{align}\label{jacobian_appendix}
    J = \begin{bmatrix} \frac{\partial \Theta}{\partial S} & \frac{\partial \Theta}{\partial P} \\ \frac{\partial T^2}{\partial S} & \frac{\partial T^2}{\partial P} \end{bmatrix} \, ,
\end{align}
where we repeat the expressions for $\Theta$ and $T^2 = T^{\mu \nu} T_{\mu \nu}$ that were given in Section \ref{sec:generalities},
\begin{align}\label{theta_Tsq_last_appendix}
    \Theta &= 4 \left( \mathcal{L}  - P\mathcal{L}_P - S\mathcal{L}_S \right) \, , \nonumber \\
    T^2 &= 4\left( S^2 + P^2 \right) \mathcal{L}_S^2 + 4\left( \mathcal{L} - P \mathcal{L}_P - S \mathcal{L}_S \right)^2 \, .
\end{align}
It is straightforward to compute the four elements of the Jacobian matrix $J$ by taking derivatives of (\ref{theta_Tsq_last_appendix}) with respect to $S$ and $P$,
\begin{align}
    \tensor{J}{^\Theta_S} &= \frac{d \Theta}{d S} = - 4 \left( P \mathcal{L}_{SP} + S \mathcal{L}_{SS} \right) \, , \nonumber \\
    \tensor{J}{^\Theta_P} &= \frac{d \Theta}{d P} = - 4 \left( P \mathcal{L}_{PP} + S \mathcal{L}_{SP} \right) \, , \nonumber \\
    \tensor{J}{^{T^2}_S} &= \frac{d T^2}{d S} = 8 S \mathcal{L}_S^2 + 8 \left( S^2 + P^2 \right) \mathcal{L}_S \mathcal{L}_{SS} - 8 \left( \mathcal{L} - P \mathcal{L}_P - S \mathcal{L}_S \right) \left( P \mathcal{L}_{SP} + S \mathcal{L}_{SS} \right) \, , \nonumber \\
    \tensor{J}{^{T^2}_P} &= \frac{d T^2}{d P} = 8 P \mathcal{L}_S^2 + 8 \left( S^2 + P^2 \right) \mathcal{L}_S \mathcal{L}_{SP} - 8 \left( \mathcal{L} - P \mathcal{L}_P - S \mathcal{L}_S \right) \left( P \mathcal{L}_{PP} + S \mathcal{L}_{SP} \right) \, .
\end{align}
We can then write out the Jacobian determinant explicitly:
\begin{align}\label{jacobian_general_appendix}
    &\det \left( J \right) = \tensor{J}{^\Theta_S} \tensor{J}{^{T^2}_P} - \tensor{J}{^\Theta_P} \tensor{J}{^{T^2}_S} \nonumber \\
    &\quad = - 32 \left( P \mathcal{L}_{SP} + S \mathcal{L}_{SS} \right) \left( P \mathcal{L}_S^2 + \left( S^2 + P^2 \right) \mathcal{L}_S \mathcal{L}_{SP} - \left( \mathcal{L} - P \mathcal{L}_P - S \mathcal{L}_S \right) \left( P \mathcal{L}_{PP} + S \mathcal{L}_{SP} \right) \right) \nonumber \\
    &\qquad + 32 \left( P \mathcal{L}_{PP} + S \mathcal{L}_{SP} \right) \left( S \mathcal{L}_S^2 + \left( S^2 + P^2 \right) \mathcal{L}_S \mathcal{L}_{SS} - \left( \mathcal{L} - P \mathcal{L}_P - S \mathcal{L}_S \right) \left( P \mathcal{L}_{SP} + S \mathcal{L}_{SS} \right) \right) \, . 
\end{align}
For a generic theory of non-linear electrodynamics, the function $\mathcal{L} ( S, P )$ will not satisfy any particular differential equation relating its derivatives with respect to $S$ and $P$, and the Jacobian determinant (\ref{jacobian_general_appendix}) will be non-vanishing.

However, for a theory of non-linear electrodynamics, the Lagrangian satisfies the partial differential equation (\ref{EM_duality_pde}). As we used in Appendix \ref{app:inductive_proof} above, this duality-invariance condition also implies constraints on the second derivatives of the Lagrangian, which are obtained by differentiating the constraint (\ref{EM_duality_pde}) with respect to $S$ and $P$. These additional relations were presented in equation (\ref{first_deriv_constraints}), which we repeat for convenience:
\begin{align}\label{deriv_constraints_last_appendix}
    \mathcal{L}_S \mathcal{L}_{SS} - \frac{1}{P} \mathcal{L}_S \mathcal{L}_P - \frac{S}{P} \left( \mathcal{L}_{SS} \mathcal{L}_P + \mathcal{L}_S \mathcal{L}_{SP} \right) - \mathcal{L}_P \mathcal{L}_{SP} &= 0 \, , \nonumber \\
    \mathcal{L}_S \mathcal{L}_{SP} + \frac{S}{P^2} \mathcal{L}_S \mathcal{L}_P - \frac{S}{P} \left( \mathcal{L}_{SP} \mathcal{L}_P + \mathcal{L}_S \mathcal{L}_{PP} \right) - 2 \mathcal{L}_P \mathcal{L}_{PP}  &= 0 \, .
\end{align}
After substituting the constraints (\ref{EM_duality_pde}) and (\ref{deriv_constraints_last_appendix}) into the expression (\ref{jacobian_general_appendix}) for the determinant and simplifying, one finds
\begin{align}
    \det \left( J \right) = 0 \, , 
\end{align}
which means that this change of variables is singular.

We have therefore shown that, in a theory of non-linear electrodynamics, there is a functional relation between the two invariants $\Theta$ and $T^2$ that can be constructed from the stress tensor. Of course, it immediately follows that any other pair of independent Lorentz scalars constructed from the stress tensor will also be dependent in such theories. For instance, in the main text of this paper we have sometimes parameterized functions of the stress tensor in terms of the two scalars $\left( \Theta, \widehat{T}^2 \right)$, where $\widehat{T}$ is the traceless part of the stress tensor, rather than in terms of $\left( \Theta, T^2 \right)$. The same conclusion $\det \left( J \right) = 0$ applies to the change of variables from $(S, P)$ to $( \Theta, \widehat{T}^2 )$, or indeed to any other two variables
\begin{align}
    X_1 \left( T_{\mu \nu} \right) \, , \quad X_2 \left( T_{\mu \nu} \right) \, .
\end{align}
To see this, we can simply enact a change of variables from $(S, P)$ to $(X_1, X_2)$ in two steps,
\begin{align}\label{combined_transformation}
    ( S , P ) \to \left( \Theta (S, P) , T^2 ( S, P ) \right) \to \left( X_1 ( \Theta, T^2 ) , X_2 ( \Theta, T^2 ) \right) \, .
\end{align}
The Jacobian for the combined transformation (\ref{combined_transformation}) is then given by the product
\begin{align}
    J \Big[ (S, P) \to (X_1, X_2) \Big] = J \Big[ (S, P) \to (\Theta, T^2) \Big] \cdot J \Big[  (\Theta, T^2) \to ( X_1 , X_2 ) \Big] \, ,
\end{align}
and by the property $\det \left( A B \right) = \det \left( A \right) \det \left( B \right)$ of determinants,
\begin{align}\label{two_dets_product}
    \det \left\{ J \Big[ (S, P) \to (X_1, X_2) \Big] \right\} = \det \left\{ J \Big[ (S, P) \to (\Theta, T^2) \Big] \right\} \cdot \det \left\{ J \Big[  (\Theta, T^2) \to ( X_1 , X_2 ) \Big] \right\} \, .
\end{align}
But we have already seen that the first determinant on the right side of equation (\ref{two_dets_product}) vanishes, so the Jacobian determinant for the combined change of variables also vanishes. Therefore any two Lorentz scalars $X_1, X_2$ constructed from the energy-momentum tensor of a theory of duality-invariant electrodynamics are functionally dependent.

\begin{footnotesize}

\end{footnotesize}

\end{document}